\documentclass[11pt]{article}

\usepackage{fullpage}

\usepackage[usenames,dvipsnames]{xcolor}
\usepackage[colorlinks,citecolor=blue,linkcolor=BrickRed]{hyperref}
\usepackage{makeidx} 
\usepackage{algorithmic}
\usepackage{graphicx,tipa}
\usepackage{arcs,lmodern,fix-cm}
\usepackage{times}
\usepackage{amsfonts,latexsym,graphicx,epsfig,amssymb,color}
\usepackage{mathdots,amsmath,amsthm,amstext,setspace}
\usepackage{amsmath,amstext,amsthm,url,setspace, enumerate}
\usepackage{slashbox,multirow}
\usepackage{rotating}
\usepackage{verbatim}

\usepackage{float}
\usepackage{graphicx}   
\usepackage{xcolor}
\usepackage{microtype}
\usepackage{soul}
\usepackage{algorithm2e}
\usepackage{tabularx}
\usepackage{placeins}
\usepackage{mathtools}
\usepackage{comment}

\SetKwBlock{Repeat}{repeat}{end}

\newtheorem{theorem}{Theorem}[section]

\newtheorem{lemma}[theorem]{Lemma}

\theoremstyle{definition}

\newtheorem{definition}[theorem]{Definition}

\newcommand{\ignore}[1]{}

\newcommand{\cost}[1]{\mathcal E_{#1}}
\def\FS{\textsc{FS}$_p$}
\newcommand{\cra}[2]{\mathcal{C}^{#1}_{#2}}

\newenvironment{showproof}{}{}

\newenvironment{showappendix}{}{}
\excludecomment{showappendix}

\begin{document}
\title{\bf 
Probabilistically Faulty Searching \\
on a Half-Line
\footnote{
This is full version of the paper with the same title which will appear in the proceedings of the 
14th Latin American Theoretical Informatics Symposium (LATIN’20), São Paulo, Brazil, May 25-29, 2020.}
}

\ignore{
\author{
Anthony Bonato\inst{1}\thanks{Research supported in part by NSERC.}
\and
Konstantinos Georgiou\inst{1}$^\star$
\and
Calum MacRury\inst{2}\thanks{Research supported by a NSERC USRA held at Ryerson University, Department of Mathematics.}
\and
Pawe\l{} Pra\l{}at\inst{1}$^\star$
\ignore{
First Author\inst{1}
\and
Second Author\inst{2,3}
\and
Third Author\inst{3}
} 
}
}

\author{
Anthony Bonato\footnotemark[2]~~\footnotemark[4] 
\and
Konstantinos Georgiou\footnotemark[2]~~\footnotemark[4]
\and
Calum MacRury\footnotemark[3]~~\footnotemark[5]
\and 
Pawe\l{} Pra\l{}at\footnotemark[2]~~\footnotemark[4]
}

\def\thefootnote{\fnsymbol{footnote}}
\footnotetext[3]{Research supported in part by NSERC of Canada.}
\footnotetext[5]{Research supported by a NSERC USRA held at Ryerson University, Department of Mathematics.}
\footnotetext[2]{
Dept. of Mathematics, Ryerson University, Toronto, Canada, \\ \texttt{
$\{$abonato,konstantinos,pralat$\}$@ryerson.ca}
}
\footnotetext[3]{
Department of Computer Science, University of Toronto,
Toronto, ON, M5S 2E4, Canada,
\texttt{cmacrury@cs.toronto.edu}
}

\maketitle

\begin{abstract}
We study \textit{$p$-Faulty Search}, a variant of the classic cow-path optimization problem, where a unit speed robot searches the half-line (or $1$-ray) for a hidden item. The searcher is probabilistically faulty, and detection of the item with each visitation is an independent Bernoulli trial whose probability of success $p$ is known. The objective is to minimize the worst case expected detection time, relative to the distance of the hidden item to the origin. A variation of the same problem
was first proposed by Gal~\cite{gal1980search} in 1980.
Alpern and Gal~\cite{alpern2003theory} proposed a so-called monotone solution for searching the line ($2$-rays); that is, a trajectory in which the newly searched space increases monotonically in each ray and in each iteration. Moreover, they conjectured that an optimal trajectory for the $2$-rays problem must be monotone. We disprove this conjecture when the search domain is the half-line ($1$-ray). We provide a lower bound for all monotone algorithms, which we also match with an upper bound. 
Our main contribution is the design and analysis of a sequence of refined search strategies, outside the family of monotone algorithms, which we call \textit{$t$-sub-monotone algorithms}. Such algorithms induce performance that is strictly decreasing with $t$, and for all $p \in (0,1)$. The value of $t$ quantifies, in a certain sense, how much our algorithms deviate from being monotone, demonstrating that monotone algorithms are sub-optimal when searching the half-line.


\vspace{0.5cm}
\end{abstract}

\section{Introduction}
\label{sec: intro}

The problem of searching for a hidden item in a specified continuous domain dates back to the early 1960's and to the early works of Beck~\cite{beck1964linear} and Bellman~\cite{bellman1963optimal}. In its simplest form, a unit speed robot (that is, a mobile agent) starts at a known location, the origin, in a known search-domain. An item, sometimes called the \emph{treasure} or the \emph{exit}, is located (hidden) at an unknown distance $d$ away from the origin, and it can be located by the robot only if it walks over it. What is the robot's trajectory that minimizes the worst case relative time that the treasure is located, compared to $d$? This worst case measure of efficiency is known as the competitive ratio of the trajectory. Interestingly, numerous variations of the problem admit trajectories inducing constant competitive ratios. In certain cases, for example, in the so-called linear-search problem where the domain is the line, tight lower bounds are known that require elaborate arguments.

We consider \textit{$p$-Faulty Search} (\FS), a probabilistic version of the classic linear-search problem in which the hidden item lies in a half-line (or $1$-ray), and the item is detected with constant probability $p$ (with independent Bernoulli trials) every time the robot walks over the item. This is a special case of a problem first proposed by Gal~\cite{gal1980search}, where the search-domain is the line (or $2$-rays). Natural solutions to the problem are so-called cyclic and monotone search patterns; that is, trajectories that process each direction periodically and where the searched space in each direction expands monotonically. In~\cite{alpern2003theory}, Alpern and Gal proposed such a solution for searching $2$-rays and they conjectured that an optimal trajectory must be cyclic and monotone. Angelopoulos in~\cite{Angelopoulos15} extended the upper bound results using cyclic and monotone trajectories for searching $m$-rays. We prove that monotone trajectories are sub-optimal for searching a $1$-ray. We do so first by establishing a lower bound for all monotone algorithms to the problem (which we also match with an upper bound), and second by designing a sequence of non-monotone trajectories inducing increasingly better performance (and deviating increasingly from being monotone).

\subsection{Related Work}
\label{sec: related work}

Search-type problems are concerned with finding a specific type of information placed within a well specified discrete or continuous domain. As a topic, it spans various sub-fields of Theoretical Computer Science and has given rise to a number of book-length treatments~\cite{ahlswede1987search,alpern2003theory,CGKMAC19,stone1975theory}. Applications range from data structures and mobile agent computing, to foraging and evolution, among others, for example, see~\cite{AH00,B05,kagan2015search,koutsoupias1996searching,schwefel1993evolution}.

The problem of searching for a hidden item in one-dimensional domains was first proposed more than 50 years ago by Beck~\cite{beck1964linear} and Bellman~\cite{bellman1963optimal} in a Bayesian context. In the 1990's, solutions to basic problem's variations were rediscovered, for example, see~\cite{baezayates1993searching,kao1996searching}. Since then, several studies of various search-type problems have resulted in an extensive literature.
Below we give representative and selective examples, with an attempt to cite relatively recent results. Variations of search-type problems that share many similarities
range from the type of search domain (for example, 1 or 2-dimensional~\cite{feinerman2012collaborative,jez2009two}, $d$-dimensional grid~\cite{cohen2017exploring},
 cycle~\cite{pattanayak2017evacuating}, polygons~\cite{CKKNOS}, graphs~\cite{angelopoulos2019expanding}, grid~\cite{BrandtUW18}, $m$-rays~\cite{BrandtFRW17}), to the number of searchers (1 or more~\cite{LMS}),
to the criterion for termination (for example, search, evacuation~\cite{Watten2017}, priority evacuation~\cite{asymtotic18}, fetching~\cite{dmtcs5528})
to the communication model (for example, wireless or face-to-face~\cite{CGGKMP}) to the type of the objective (for example, minimize worst case or average case~\cite{ChuangpishitGS18}) to cost specs (for example, turning costs~\cite{demaine2006online}, cost for revisiting~\cite{Bose16}), to the measure of efficiency (for example, time, energy~\cite{CzyzowiczGKKKLN19}) to the knowledge of the input (none or partial~\cite{Bose13}) and to other robots' specs (for example, speeds~\cite{CzyzowiczKKNOS17}, faults~\cite{GKLPP19}, memory~\cite{Reingold}), just to name a few. 
More recently, Fraigniaud et al. considered in~\cite{fraigniaud2019parallel} a Bayesian search problem in a discrete space, where a set of searchers are trying to locate a treasure placed, according to some distribution, in one of the boxes indexed by positive integers.
Since it is outside the scope of this work to provide a comprehensive list of the large related literature, we further refer the interested reader to \cite{alpern2003theory,Alpern2013,CGK19search,GAL}.

The version of linear search that we study, where the searcher is probabilistically faulty, was presented as an open problem by Gal in~\cite{gal1980search}. Later in~\cite{alpern2003theory} (see chapter 8.6.2), Alpern and Gal provided a search strategy when the search domain is a line. In particular, they considered \textit{cyclic} search trajectories where the robot alternates between searching each of the two directions, and each time monotonically increasing the searched space. Among the same family of algorithms that moreover expand the searched space in each direction geometrically, the authors provided the optimal trajectory. In addition, they conjectured that cyclic and monotone trajectories are in fact optimal. Along the same lines, \cite{Angelopoulos15} studied cyclic and monotone trajectories for searching $m$-rays. In a variation of the problem where the hidden item detections are not Bernoulli trials, \cite{Angelopoulos15} showed also that cyclic trajectories are in fact sub-optimal. For this and many other variations of probabilistically searching, where the probability of success is not known, optimal strategies remain open.

\subsection{Main Contributions \& Paper Organization}
\label{sec: contributions and organization}

We introduce and study $p$-Faulty Search (\FS), a variation of the classic linear-search (cow-path) problem, in which the search space is the half-line, and detection of the hidden item (treasure) happens with known probability $p$. We are interested in designing search strategies that induce small competitive ratio, as a function of $p$; that is, that minimize the worst case expected detection time of the hidden item, with respect to its placement $d$, relative to the optimal performance of an algorithm that knows in advance the location of the item (so we normalize the expected performance both by $d$ and $p$).

We focus on two families of search algorithms, which indicate that optimal solutions to \FS\ may be particularly challenging to find.
First, we study a natural family of algorithms, that we call monotone algorithms, which intuitively are determined by non-decreasing turning points $x_i$ where searcher returns to the origin before expanding the searched space. Given that turning points increase geometrically; that is, when $x_i = b^i$, relatively straightforward calculations determine the optimal expansion factor $b=b(p)$. In fact, a simplified argument shows that in the cow-path problem (that is, when the search space consists of $2$-rays and $p=1$) the optimal expansion factor is $b=2$. A more tedious argument (and one of our technical contributions), as in the cow-path problem, shows that the aforementioned choice of geometrically increasing $x_i$'s for \FS\ is in fact optimal among the family of monotone algorithms. Our main technical contribution pertains to the design and analysis of a family of algorithms that we call $t$-sub-monotone, which provide a sequence of refined search strategies which induce competitive ratios that strictly decrease with $t$, for every $p \in (0,1)$. Somehow surprisingly, our findings show that plain-vanilla, and previously considered, algorithms for \FS\ are sub-optimal.

The organization of our paper is as follows. In Section~\ref{sec: problem definition}, we define problem \FS\ formally, we introduce measures of efficiency and we complement with preliminary and important observations. Section~\ref{sec: back to origin} studies the special family of monotone search algorithms. In particular, in Section~\ref{sec: back to origin} we propose and analyze a specific monotone algorithm where turning points increase geometrically. Section~\ref{sec: back to origin lower bound} contains one of our technical contributions, in which we prove that the monotone algorithm presented in the previous section is in fact optimal within the family. Our main technical contribution is in Section~\ref{sec: sub origin}, which introduces and studies the family of $t$-sub-monotone algorithms. Performance analysis of the family of algorithms is presented in Section~\ref{sec: performance analysis}. In Section~\ref{sec: how to choose good suborigin algo}, we propose a systematic method for choosing parameters for the $t$-sub-monotone algorithm with the objective to minimize their competitive ratio. Our formal findings are evaluated in Section~\ref{sec: t suborigin algo, t<=10}, where we demonstrate the sequence of strictly improved competitive ratios by $t$-sub-monotone algorithms when $t\leq 10$. As our proposed parameters for the algorithms are obtained as the roots to high degree ($\Theta(t)$) polynomials, are results, for the most part, cannot be described by closed formulas. However, in Section~\ref{sec: heuristics}, we selectively discuss heuristic choices of the parameters that induce nearly optimal search strategies and whose performance can be quantified by closed formulas. We also quantify formally the boundaries of $t$-sub-monotone algorithms, and we show that the competitive ratio of our $10$-sub-monotone is off additively by at most $10^{-6}$ from the best performance we can achieve by letting $t$ grow arbitrarily.
In the final section, we conclude with open problems.\begin{showappendix}
Due to space limitations, any omitted proofs can be found in the Appendix. 
\end{showappendix}

\section{Problem Definition and Preliminary Observations}
\label{sec: problem definition}

In \textit{$p$-Faulty Searching on a Halfline} (\FS) a speed-1 searcher (or robot) is located at the origin of the infinite half-line. At unknown distance $d$ bounded away from the origin, which bound we set arbitrarily to 1, 
there is an item (or treasure) which is located/detected by the robot with constant and known probability $p$ every time the robot passes over it (that is, detection trials are mutually independent and each has probability of success $p$). Also, for the sake of simplifying the analysis, we assume that the probability of detection becomes 1 if the treasure is placed exactly at a point where the robot changes direction.
As we will see later, the worst placements of the treasure will be proven to be arbitrarily close to the turning points.

Given a robot's trajectory $T$, probability $p$ and distance $d$, the \textit{termination time} $\cost{T}(d)$ is defined as the expected time that the robot detects the treasure for the first time. Feasible solution to \FS\ are robot's trajectories that induce bounded termination time (as a function of $p,d$) for all $p\in (0,1)$ and for all $d\geq 1$.

Note that $p$ is part of the input to an algorithm for \FS, while $d$ is unknown. Hence, trajectories may depend on $p$ but not on $d$. It is also evident that for a robot's trajectory to induce bounded termination time for all treasure placements, the robot needs to visit every point of the half-line, past point 1, infinitely many times. 
As it is also common in competitive analysis, we measure the performance of a search strategy relative to the optimal offline algorithm; that is, an algorithm that knows where the treasure is. Since such an algorithm needs to travel for time $d$ to reach the treasure, as well as one would need $1/p$ trials, in expectation, before detecting it, we are motivated to introduce the following measure of efficiency for search trajectories.

\begin{definition}
The \textit{competitive ratio} of search strategy $T$ for \FS\ is defined as
$
\cra{T}{p}:=\sup_{d\geq 1} \left\{ \tfrac {p\cost{T}(d)}{d} \right\}.
$
\end{definition}

Trajectory solutions (or search strategies) to problem \FS\ are in correspondence with infinite sequences $\{t_i\}_{i\geq 0}$ of turning points, satisfying $t_0=0$, $t_i\geq 0$, $t_{2i+1}>t_{2i}$ and $t_{2i}<t_{2i-1}$, for all $i\geq 0$. Indeed such a sequence $\{t_i\}_{i\geq 0}$ corresponds to the trajectory in which robot moves from $t_{2i}$ to $t_{2i+1}$ (moving away from the origin), and from $t_{2i-1}$ to $t_{2i}$ (moving toward the origin), each time changing direction of movement, where $i=1, 2, \ldots$.

For search strategy $T$ and treasure location $d$ (except from the turning points of $T$), let $f_i$ denote the time till the robot passes over the treasure for the $i$'th time. Since the probability of successfully detecting the treasure is $p$, we have 
$
\cost{T}(d)=\sum_{i=1}^\infty p(1-p)^{i-1}f_i.
$
In what follows, we express the expected termination time with respect to the additional time between two visitations of the treasure.

\begin{lemma}
\label{lem: exp term incr}
Let $f_0=0$, and let $g_i=f_i-f_{i-1}$. We then have that $\cost{T}(d)=\sum_{i=1}^\infty (1-p)^{i-1}g_i.$
\end{lemma}
\begin{showproof}
\begin{proof}
\ignore{
Note that for each $i$ we have $f_i=\sum_{j=1}^ig_i$. Hence,
\begin{align*}
\cost{T}(d)
=& \sum_{i=1}^\infty p(1-p)^{i-1}f_i \\
=& p \sum_{i=1}^\infty (1-p)^{i-1}\sum_{j=1}^ig_i \\
=& p \sum_{i=1}^\infty g_i \sum_{j=i-1}^\infty(1-p)^{j-1} \\
=& p \sum_{i=1}^\infty (1-p)^{i-1}g_i \sum_{j=0}^\infty(1-p)^{j} \\
=&\sum_{i=1}^\infty (1-p)^{i-1}g_i.
\end{align*}
}
Note that for each $i$ we have $f_i=\sum_{j=1}^ig_j$. We then have that
\begin{align*}
\cost{T}(d)
=& \sum_{i=1}^\infty p(1-p)^{i-1}f_i \\
=& p \sum_{i=1}^\infty (1-p)^{i-1}\sum_{j=1}^ig_j \\
=& p \sum_{j=1}^\infty g_j \sum_{i=j}^\infty(1-p)^{i-1} \\
=& p \sum_{j=1}^\infty (1-p)^{j-1}g_j \sum_{i=0}^\infty(1-p)^{i} \\
=&\sum_{j=1}^\infty (1-p)^{j-1}g_j,
\end{align*}
and the proof follows.
 \end{proof}
\end{showproof}


\section{Monotone Trajectories}
\label{sec: back to origin}

We explore the simplest possible trajectories for \FS\ in which the searcher repeatedly returns to the origin every time she changes direction during exploration and before exploring new points in the half-line. More formally, \textit{monotone trajectories} for \FS\ are search algorithms $T=\{t_i\}_{i\geq 1}$, defined as\footnote{Alternatively, we could have defined monotone trajectories so as to return to location 1, instead of the origin, since we know that $d\geq 1$. Our analysis next shows that such a modification would not improve the competitive ratio.}
$
t_{2i}=0, ~~t_{2i+1}=x_i, ~~i=1,2,\ldots,
$
where $\{x_i\}_{i \geq 1}$ is a strictly increasing sequence with $x_i \rightarrow \infty$. Note that, in particular, we allow $x_i = x_i(p)$.
The present section is devoted into determining the best monotone algorithm for \FS. More specifically, we prove the following.

\begin{theorem}
\label{thm: back to origin optimal}
The optimal monotone algorithm for \FS\ has competitive ratio $\frac{4+4\sqrt{1-p}}{2-p}-p$.
\end{theorem}

The proof of Theorem~\ref{thm: back to origin optimal} is given in the next two sections. In Section~\ref{sec: back to origin upper bound} we propose a specific monotone algorithm with the aforementioned performance (see Lemma~\ref{lem: upper bound backtoorigin}), while in Section~\ref{sec: back to origin lower bound} we show that no monotone algorithm performs better (see Lemma~\ref{lem: lower bound backtoorigin}). Somewhat surprisingly we show in Section~\ref{sec: sub origin} that the upper bound of Theorem~\ref{thm: back to origin optimal} is in fact sub-optimal.

\subsection{An Upper Bound Using Monotone Trajectories}
\label{sec: back to origin upper bound}

In this section we propose a specific monotone algorithm with the performance promised by Theorem~\ref{thm: back to origin optimal}. 
In particular, we consider ``restricted'' trajectories determined by increasing sequences $\{x_i\}_{i \geq 1}$, where $x_i=b^i$ and $b=b(p)>1$. Within this sub-family, we determine the optimal choice of $b$ that induces the smallest competitive ratio. For this, we first determine the placements of the treasure that induce the worst competitive ratio,  given a search trajectory. As stated before, in the following analysis we make the assumption that the treasure is not placed at any turning point.

\begin{lemma}
\label{lem: worst placement of treasure, backtoorigin}
Consider a monotone algorithm $T$, determined by the strictly increasing sequence $\{x_i\}_{i \geq 1}$. If the treasure appears in interval $(x_r, x_{r+1})$, then
the competitive ratio is no more than
$$
2\frac p{x_r} \sum_{i=1}^r x_i
+ 2 \frac p{x_r} \sum_{i\geq 1}(1-p)^{2i-1}x_{r+i}
+  \frac{p^2}{2-p}.
$$
\end{lemma}
\begin{showproof}
\begin{proof}
Suppose that the treasure is located at point $d=x_r + y \in (x_r, x_{r+1})$, where $0<y<x_{r+1}-x_r$. With that notation in mind (see also Figure~\ref{fig: BackToOrigin}), we compute the time intervals $g_i$ between consecutive visitations, as they were defined in Lemma~\ref{lem: exp term incr}. We have that
\begin{align*}
g_1 &= 2\sum_{i=1}^r x_i + x_r +y = 2\sum_{i=1}^r x_i + d\\
g_{2i} &= 2(x_{r+i}-x_r-y)= 2(x_{r+i}-d), \qquad ~~i=1,\ldots, \infty \\
g_{2i+1} &= 2x_r+2y=2d, \qquad ~~i=1,\ldots, \infty.
\end{align*}

\begin{figure}[!ht]
\vspace{-0.1in}
                \centering
                \includegraphics[scale=0.7]{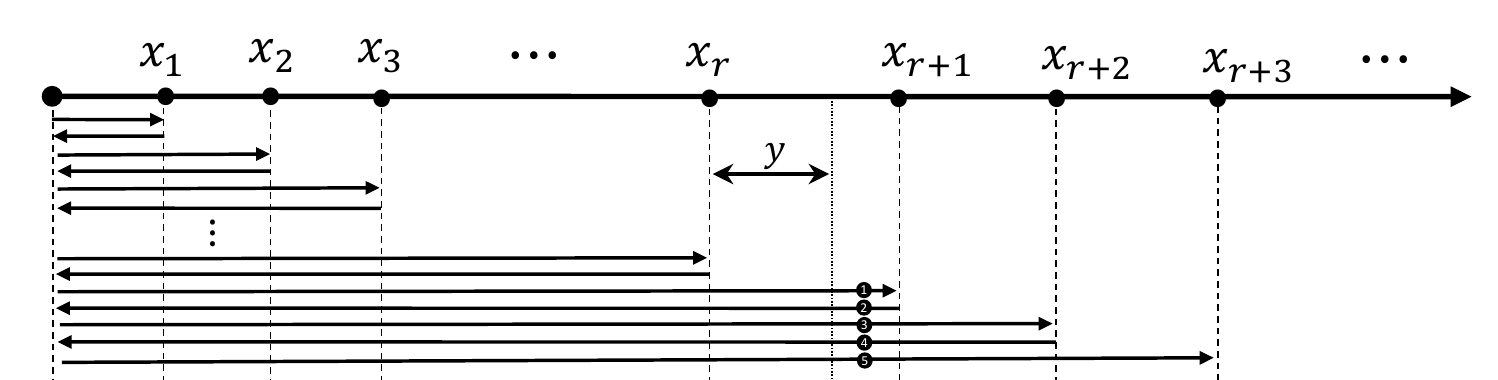}
                \vspace{-0.1in}
\caption{Monotone algorithm $\{x_i\}_{i\geq 1}$. Figure also depicts the first 5 visitations of the treasure that is placed at  $x_r+y$.
}
                \label{fig: BackToOrigin}
\vspace{-0.1in}
\end{figure}

Therefore, by Lemma~\ref{lem: exp term incr} the expected termination time $\cost{T}(d)$ for algorithm $T$ is 
\begin{align*}
\sum_{i=1}^\infty (1-p)^{i-1}g_i
&= g_1
+ \sum_{i\geq 1}(1-p)^{2i-1}g_{2i}
+ \sum_{i\geq 1}(1-p)^{2i}g_{2i+1} \\
&=
\left( 2\sum_{i=1}^r x_i + d \right)
+
2\left( \sum_{i\geq 1}(1-p)^{2i-1}(x_{r+i}-d)
 \right)
+
2d\left( \sum_{i\geq 1}(1-p)^{2i}
 \right) \\
&=
2 \sum_{i=1}^r x_i
+ 2 \sum_{i\geq 1}(1-p)^{2i-1}x_{r+i}
+ d \left( 1 -2p \sum_{i\geq 1}(1-p)^{2i-1} \right)\\
&=2 \sum_{i=1}^r x_i
+ 2 \sum_{i\geq 1}(1-p)^{2i-1}x_{r+i}
+ d \frac{p}{2-p}.
\end{align*}
Recall that the competitive ratio of this algorithm is $p \cost{T}(d)/d$, and hence, in the worst case, $d$ approaches $x_r$ from the right.
~\qed
 \end{proof}
\end{showproof}

We are now ready to prove the promised upper bound. 
\begin{lemma}
\label{lem: upper bound backtoorigin}
The monotone trajectory $T=\{x_i\}_{\geq 1}$, where $x_i=b^i$ and
$b:=\frac{1}{\sqrt{1-p}\left(2-p-\sqrt{1-p}\right)}$
has competitive ratio $\frac{4+4\sqrt{1-p}}{2-p}-p$.
\end{lemma}

\begin{showproof}
\begin{proof}
We study the restricted family of monotone trajectories $T=\{x_i\}_{\geq 1}$, where $x_i=b^i$, for some $b=b(p)$.
By Lemma~\ref{lem: worst placement of treasure, backtoorigin}, the competitive ratio of search strategy $T$ is at most
\begin{align}
\sup_r \left\{
2\frac p{b^r} \sum_{i=1}^r b^i
+ 2 \frac p{b^r} \sum_{i\geq 1}(1-p)^{2i-1}b^{r+i}
+  \frac{p^2}{2-p}
\right\}
=&
\sup_r \left\{
p\frac{2b \left(b^r-1\right)}{b^{r}(b-1)}
+
p\frac{2b (1-p)}{1-b (1-p)^2}
+
\frac{p^2}{2-p}\right\} \notag \\
=&\lim_{r\rightarrow \infty} \left\{
p\frac{2b \left(b^r-1\right)}{b^{r}(b-1)}
+
p\frac{2b (1-p)}{1-b (1-p)^2}
+
\frac{p^2}{2-p}\right\} \notag \\
=&
p\frac{2b}{b-1}
+
p\frac{2b (1-p)}{1-b (1-p)^2}
+
\frac{p^2}{2-p}. \label{equa: comp ratio wrt b origin}
\end{align}
Calculations above assume that 
$b<1/(1-p)^2$, as otherwise, the second summation is divergent.
We will make sure later that our choice of $b$ complies with this condition.
Note also that for $x_i$ to be increasing, we need $b>1$.
Now, denote expression~\eqref{equa: comp ratio wrt b origin} by $f(b)$. We will determine the choice of $b$ that minimizes $f(b)$, given that $1<b<1/(1-p)^2$.

It is straightforward to see that $\frac{d^2}{db^2}f(b)=4p \left(\frac{(1-p)^3}{1-\left(b (p-1)^2\right)^3}+\frac{1}{(b-1)^3}\right)$, and hence, $f(b)$ is convex when $b\in \left(1, 1/(1-p)^2\right)$. Hence, if $\frac{d}{db}f(b)$ has a root in $\left(1, 1/(1-p)^2\right)$, that would be a minimizer. Indeed, $$\frac{d}{db}f(b)=2p\left(\frac{1-p}{\left(1-b (1-p)^2\right)^2}-\frac{1}{(b-1)^2}\right)$$ has two roots $\frac{1}{\sqrt{1-p}\left(\pm(2-p)-\sqrt{1-p}\right)}$, one being positive and one negative (for all values of $p\in (0,1)$). We choose the positive root, that we call $b_p$, 
and it is elementary to see that $1<b_p<1/(1-p)^2$, for all $p\in (0,1)$, as wanted. Substituting $b=b_p$ in~\eqref{equa: comp ratio wrt b origin} gives the competitive ratio promised by the statement of the lemma. 
\end{proof}
\end{showproof}
\ignore{
cr[p_, b_] := 2*b/(b - 1) + 2*b*(1 - p)/(1 - b*(1 - p)^2) + p/(2 - p)
FullSimplify[ cr[p, 1/(-1 + Sqrt[-(-2 + p)^2 (-1 + p)] + p)] ,
 Assumptions -> p > 0 && p < 1]
}


\subsection{Lower Bounds for Monotone Trajectories}
\label{sec: back to origin lower bound}

This section is devoted to proving the following lemma.
\begin{lemma}
\label{lem: lower bound backtoorigin}
Every monotone trajectory has competitive ratio at least $\frac{4+4\sqrt{1-p}}{2-p}-p$.
\end{lemma}

Consider an arbitrary monotone algorithm $T=\{f_i\}_{i\geq 0}$, where $f_i$ is a monotone sequence tending to infinity, and which determines the turning points of the algorithm.
Without loss of generality, 
we set $f_0=1$, as otherwise we may scale all turning points by $f_0$.  Our lower bound will be obtained by restricting the placement of the treasure arbitrary close to (and $\epsilon>0$ away after) turning points $f_k$ (this may only result in a weaker lower bound). Taking $\epsilon \rightarrow 0$, we obtain that 
\begin{align*}
&g_1^k=2\sum_{i=0}^k{f_i}+f_k, \\
&g_{2i}^k=2(f_{k+i}-f_k), \\
&g_{2i+1}^k=2f_k,
\end{align*}
where the superscript $k$ of $g_i^k$ indicates exactly 
the placement of the treasure at $f_k$.
In what follows, and for a fixed integer $\ell$, we define
$$\alpha :=\frac 12+\frac{1}{2-p}-\frac{c}{2p},
~~\beta_{i,k}:=(1-p)^{2(i-k)-1}, \textrm{for }k+1\leq i\leq \ell,
~~\gamma_{\ell,k}:=\frac{(1-p)^{2(\ell-k)+1}}{p(2-p)}.
$$

We have the following lemma.

\begin{lemma}
\label{lem: condition on alphas}
Let $c$ be the optimal competitive ratio that can be achieved by monotone trajectory $T$. For every integer $\ell$ and for every $0\leq k\leq \ell$ we have that
\begin{align}
\label{eq-2}
\sum_{i=0}^{k-1}f_i+\alpha f_{k}+\sum_{i=k+1}^{\ell}\beta_{i,k}f_i+\gamma_{\ell,k}f_{\ell}\leq 0.
\end{align}
\end{lemma}

\begin{showproof}
\begin{proof}
If the treasure is placed arbitrarily close to turning point $f_k$, then by Lemma \ref{lem: exp term incr}, a lower bound to the best possible competitive ratio $c$ satisfies the following (infinitely many) constraints: 
$$c\geq \frac{p}{f_k}\sum_{i=1}^{\infty}(1-p)^{i-1}g_i^k, \quad k=0,\ldots, \infty.$$
We next restrict our attention to the first $\ell+1$ such constraints, where $\ell$ is an arbitrary integer. Hence, we require that
$$c\geq \frac{p}{f_k}\sum_{i=1}^{\infty}(1-p)^{i-1}g_i^k, \quad k=0,\ldots,\ell.$$
Now, multiply both hand-sides of the inequalities by $f_k/p$ to obtain 
\begin{align*}
f_k\frac{c}{p}\geq \sum_{i=1}^{\infty}(1-p)^{i-1}g_i^k &=2\sum_{i=1}^k{f_i}+f_k+2\sum_{i=1}^{\infty} (1-p)^{2i-1}(f_{k+i}-f_k)+2\sum_{i=1}^{\infty}(1-p)^{2i}f_k\\
&\geq 2\sum_{i=1}^k{f_i}+f_k+2\sum_{i=1}^{\ell-k}(1-p)^{2i-1}(f_{k+i}-f_k) \\
&~~~+2\sum_{i=\ell-k+1}^{\infty}(1-p)^{2i-1}(f_\ell-f_k)+2f_k\sum_{i=1}^{\infty}(1-p)^{2i}\\
&=2\sum_{i=1}^{k-1}f_i+f_k\left(3-2\sum_{i=1}^{\infty}(1-p)^{2i-1}+2 \sum_{i=1}^{\infty}(1-p)^{2i}\right)\\
&~~~+2\sum_{i=k+1}^{\ell}(1-p)^{2(i-k)-1}f_i+2f_\ell\frac{(1-p)^{2(\ell-k)+1}}{p(2-p)}.
\end{align*}
We conclude that $f_k c / p$ is at least the last term above, so after rearranging the terms of the inequality, bringing them all on one side, and factoring out the $f_i$ terms, we have that
\begin{align*}
\sum_{i=0}^{k-1}f_i+\left(\frac{1}{2}+\frac{1}{2-p}-\frac{c}{2p}\right) f_k+\sum_{i=k+1}^{\ell}(1-p)^{2(i-k)-1}f_i+\frac{(1-p)^{2(\ell-k)+1}}{p(2-p)}f_\ell \leq 0,
\end{align*}
as desired.
 \end{proof}
\end{showproof}

Recall that $f_0=1$. Our lower bound derived in the proof of Lemma~\ref{lem: lower bound backtoorigin} is obtained by finding the smallest $c$ satisfying constraints~\eqref{eq-2}, and in particular, inducing a strictly increasing sequence of $f_i$ in $i$. Note that minimizing $c$ subject to constraints~\eqref{eq-2} in variables $f_1, \ldots, f_\ell,c$ is a non-linear program. To obtain a lower bound for $c$, we observe that the only negative coefficients of variables $f_i$ are those on the diagonal; that is, the coefficient of $f_k$ in the $k$'th constraint.
 This allows us to apply repeatedly back substitution to obtain a lower bound for all $f_i$ and hence, $c$ as well, assuming that the visiting points $f_i$ are increasing in $i$. 
 Equivalently, for the optimal $c$ that an algorithm can achieve, we may treat (for the sake of the analysis) all inequalities~\eqref{eq-2} as being tight, giving rise to the linear system

\begin{equation}
\label{equa: optimal turning points}
A_\ell f=a,
\end{equation}
in variables $f^T=(f_1, \ldots, f_\ell)$, where
\[
  A_\ell :=
  \left[ {\begin{array}{ccccc}
   \beta_{1,0}&\beta_{2,0} &\beta_{3,0}&\ldots &\gamma_{\ell,0}+\beta_{\ell,0}\\
   \alpha & \beta_{2,1} &\beta_{3,1}&\ldots & \gamma_{\ell,1}+\beta_{\ell,1} \\
   1& \alpha & \beta_{3,2}&\ldots & \gamma_{\ell,2}+\beta_{\ell,2} \\
   1&1&\alpha & \ldots & \gamma_{\ell,3}+\beta_{\ell,3} \\
     \vdots&\vdots&\vdots & \ldots &\vdots \\
      1&1&1 & \ldots & \gamma_{\ell,\ell-1}+\beta_{\ell,\ell-1} \\
   \end{array} } \right],
\ignore{
   F=
  \left[ {\begin{array}{c}
   f_1\\
   f_2\\
   f_3\\
   f_4\\
   \vdots\\
   f_{l}\\
  \end{array} } \right],
  }
 ~~ a:= \left[ {\begin{array}{c}
   -\alpha\\
   -1\\
   -1\\
   -1\\
   \vdots\\
   -1\\
  \end{array} } \right].
\]
Constraints~\eqref{equa: optimal turning points} may be thought as the defining linear system on $f_i$'s  that give the optimal turning strategies, assuming that the treasure can only be placed arbitrarily close and after any of the $\ell$ first turning points of a search trajectory. 
In other words, given that any monotone algorithm is defined by a sequence of turning points, these points can be chosen so as to minimize the competitive ratio with the assumption that the hidden item will be nearly missed after each turning point. Having the competitive ratio be independent of the treasure's placement gives a lower bound to the competitive ratio of the algorithm. 
The proof of Lemma~\ref{lem: lower bound backtoorigin} follows directly from the following technical lemma.

\begin{lemma}
\label{lem: monotonicity implies bound}
Linear system~\eqref{equa: optimal turning points}, in variables $f_i$, defines a monotone sequence of turning points only if $c\geq \frac{4+4\sqrt{1-p}}{2-p}-p$.
\end{lemma}

\begin{showproof}
\begin{proof}
We proceed by finding a closed formula for $f_{\ell-1}$ and then imposing monotonicity. Our first observation is that for all $0\leq k\leq \ell-1$ we have that $\gamma_{\ell,k}+\beta_{\ell,k}=\frac{(1-p)^{2(\ell-k)-1}}{(2-p)p}$. Setting $r:=\frac{1}{(2-p)p}$ allows us to  rewrite the matrix of system~\eqref{equa: optimal turning points} as
\[
  A_\ell =
  \left[ {\begin{array}{ccccc}
   (1-p)&(1-p)^3&(1-p)^5&\ldots &r(1-p)^{2\ell-1}\\
   \alpha & (1-p) &(1-p)^3&\ldots & r(1-p)^{2\ell-3} \\
   1& \alpha & (1-p)&\ldots & r(1-p)^{2\ell-5} \\
   1&1&\alpha & \ldots & r(1-p)^{2\ell-7} \\
     \vdots&\vdots&\vdots & \ldots &\vdots \\
      1&1&1 & \ldots & r(1-p) \\
   \end{array} } \right].
\]
We proceed by applying elementary row operations to the system. From each row of $A_\ell$ (except the last one) we subtract a $(1-p)^2$ multiple of the following row to obtain linear system $\bar{A_\ell} f = \bar{b}$, where

\[
\bar{A_\ell}=
  \left[ {\begin{array}{cccccc}
   1-p-\alpha  (1-p)^2&0 &0&\ldots &0&0\\
   \alpha-(1-p)^2 & 1-p-\alpha  (1-p)^2 &0&\ldots & 0 & 0\\
   1-(1-p)^2& \alpha-(1-p)^2 & 1-p-\alpha  (1-p)^2&\ldots & 0 & 0 \\
   1-(1-p)^2&1-(1-p)^2&\alpha-(1-p)^2 & \ldots & 0 &0 \\
     \vdots&\vdots&\vdots & \ldots &\vdots&\vdots \\
      1-(1-p)^2&1-(1-p)^2&1-(1-p)^2 & \ldots &1-p-\alpha  (1-p)^2&0 \\
        1&1&1& \ldots &\alpha&r(1-p)\\
   \end{array} } \right]
\]
and $\bar{b}^T =
(-\alpha+(1-p)^2,
 -1+(1-p)^2,
 -1+(1-p)^2,
   -1+(1-p)^2,
   \hdots,
   -1+(1-p)^2,
   -1)
   $.
Now set
$$s:=1-p-\alpha  (1-p)^2, ~~t:=\alpha-(1-p)^2, ~~w:=1-(1-p)^2,$$
and define $\ell \times \ell$ matrix
\[
\ignore{
  D_\ell :=
  \left[ {\begin{array}{ccccccc}
   s&0 &0&\ldots &0&0&-t\\
   t & s &0&\ldots &0& 0 & -w\\
 w& t & s&\ldots & 0 &0& -w \\
   w&w&t & \ldots & 0 &0&-w \\
     \vdots&\vdots&\vdots & \ldots &\vdots&\vdots&\vdots \\
      w&w&w& \ldots &s&0&-w \\
        w&w&w& \ldots &t&s&-w \\
        1&1&1& \ldots &1&\alpha&-1 \\
   \end{array} } \right],
}
 C_{\ell}:=
  \left[ {\begin{array}{ccccccc}
   s&0 &0&\ldots &0&-t&0\\
   t & s &0&\ldots &0& -w&0\\
 w& t & s&\ldots & 0 &-w&0 \\
   w&w&t & \ldots & 0 &-w&0 \\
     \vdots&\vdots&\vdots & \ldots &\vdots&\vdots&\vdots \\
      w&w&w& \ldots &s&-w&0 \\
        w&w&w& \ldots &t&-w&0 \\
        1&1&1& \ldots &1&-1&r(1-p) \\
   \end{array} } \right].
\]
By Cramer's rule we have that
$$
f_{\ell-1}=\frac{\mathrm{det}(C_{\ell})}{\mathrm{det}(A_\ell)}.
$$
Note that $\mathrm{det}(A_\ell)= \left( 1-p-\alpha  (1-p)^2 \right)^{\ell-1}r(1-p)$.

Next we compute $\mathrm{det}(C_{\ell})$. We denote the $(\ell-1)\times(\ell-1)$ principal minor of $C_\ell$ as $B_{\ell-1}$. The last row of $B_{\ell-1}$ is $(w,w,\ldots,w,t,-w)$. We further denote by $L_{\ell-1}$ the matrix we obtain from $B_{\ell-1}$ by scaling its last row by $w$ so that it reads $(1,1,\ldots,1,t/w,-1)$. Finally, we denote by $K_{\ell-1}$ the matrix we obtain by replacing the last row of $B_{\ell-1}$ by $(1,1,\ldots, 1,1,-1)$; that is, the all-1 row except from the last entry which is -1. With this notation in mind, we note that
$$
\det(C_\ell) = -r(1-p)\det(B_{\ell-1}) =- \frac{r(1-p)}{w}\det(L_{\ell-1}).
$$
Now expanding the determinants of $K_{\ell-1}, L_{\ell-1}$ with respect to their first rows we obtain the system of recurrence equations
\begin{align*}
\det(K_{\ell-1}) &= s \det(K_{\ell-2}) - w \det(L_{\ell-2}), \\
\det(L_{\ell-1}) &= s \det(K_{\ell-2}) - t \det(L_{\ell-2}).
\end{align*}
We solve the first one with respect to $\det(L_{\ell-2})$ and we substitute to the second one to obtain the following recurrence exclusively on $K_\ell$
$$
\det(K_\ell) + (t-s)\det(K_{\ell-1}) + s(w-t) \det(K_{\ell-2}) = 0.
$$
The characteristic polynomial of the latter degree-2 linear recurrence has discriminant equal to
$$
(t-s)^2 - 4s(w-t)
=
\frac{1}{4} \left((2-p)^2  c^2 +2 ((p-2) p+4) (p-2) c +p^2 ((p-4) p+12)\right),
$$
which in particular is a degree-2 polynomial $g(c)$ in the competitive ratio $c$ and has discriminant $4 (2-p)^2 (1-p)$. Since $g(c)$ is convex, we conclude that the discriminant of the characteristic polynomial is non-negative when $c$ is larger than the largest root of $g(c)$, that is when
$$
c\geq
\frac{
 (4-(2-p) p) (2-p) + 4(2-p) \sqrt{1-p}
}{(2-p)^2}
=
\frac{4+4\sqrt{1-p}}{2-p}-p,
$$
and the proof follows. 
 \end{proof}
\end{showproof}

\section{Sub-Monotone Trajectories}\label{sec: sub origin}

For a fixed integer $t$, we consider a \textit{$t$-sub-monotone trajectory} that is defined by a strictly increasing sequence $\{x_i\}_{i \geq 1}$, where $x_i = \beta^i$ for some $\beta=\beta(p)>1$, and $\{\gamma_i\}_{i=1,\ldots,t}$ (where $\gamma_i = \gamma_i(p)$) satisfying
$
1<\gamma_1<\gamma_2<\ldots <\gamma_t<\beta.
$
For convenience, we introduce abbreviations $\gamma_0=1$ and $\gamma_{t+1}=\beta$. For the formal description of the trajectory, we introduce the notion of a \textit{$t$-hop between consecutive points $x_r, x_{r+1}$}, see Algorithm~\ref{t-hop}, which is a sub-trajectory of the robot starting from $x_r$ and finishing at $x_{r+1}$.
\begin{algorithm}
\caption{$t$-Hop between $x_r, x_{r+1}$}
\label{t-hop}
\begin{algorithmic}[1]
\FOR{$j=1,\ldots,t$}
	\STATE Move from $\gamma_{j-1}x_r$ to $\gamma_{j}x_r$
	\STATE Move from $\gamma_{j}x_r$ to $\gamma_{j-1}x_r$
	\STATE Move from $\gamma_{j-1}x_r$ to $\gamma_{j}x_r$
\ENDFOR
\STATE Move from $\gamma_{t}x_r$ to $x_{r+1}$
\end{algorithmic}
\end{algorithm}
Given parameters $\gamma_i$ and $\beta$, the $t$-suborigin trajectory is defined in Algorithm~\ref{t-suborigin}.
\begin{algorithm}
\caption{$t$-Sub-Monotone Trajectory}
\label{t-suborigin}
\begin{algorithmic}[1]
\STATE Move from the origin to $x_1$, then to the origin and then to $x_1$.
\FOR{$r=1,\ldots,\infty$}
	\STATE Perform a $t$-hop between $x_r, x_{r+1}$.
	\STATE Move from $x_{r+1}$ to the origin
	\STATE Move from the origin to $x_{r+1}$
\ENDFOR
\end{algorithmic}
\end{algorithm}
The trajectory of the robot performing a $t$-sub-monotone search is depicted in Figure~\ref{fig:t-hop} that shows a $t$-hop between points $x_r$ and $x_{r+1}$. 
 \begin{figure}[!ht]
                \centering
                \includegraphics[scale=.7]{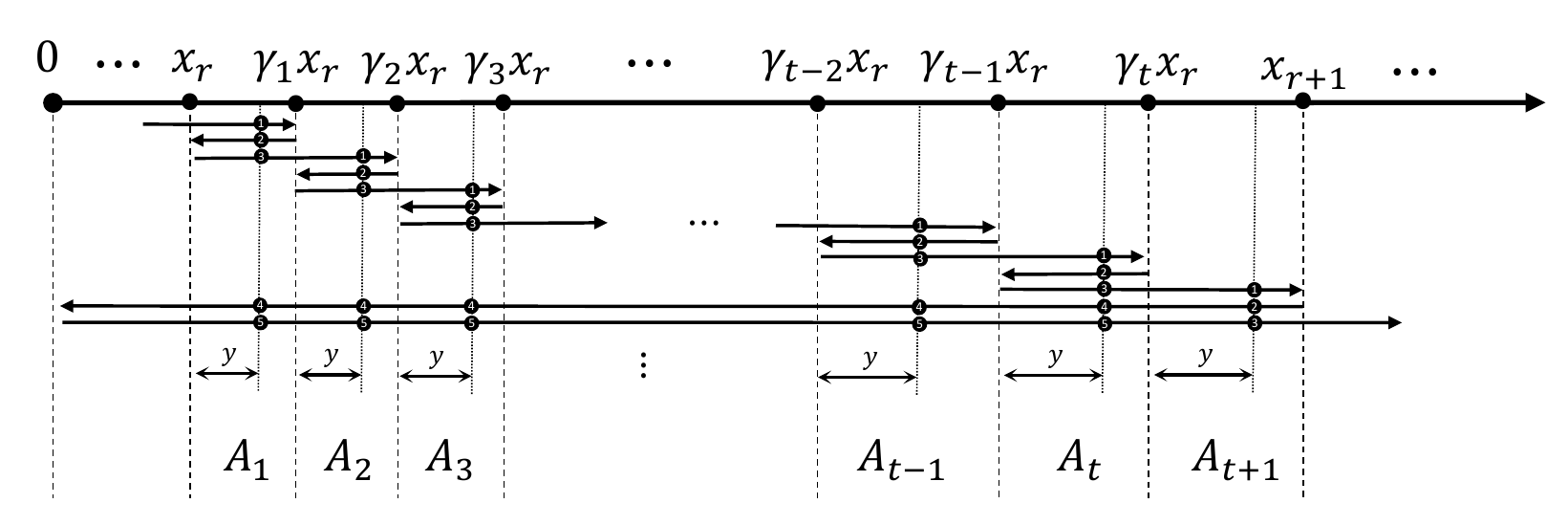}
                \vspace{-0.1in}
\caption{
$t$-sub-monotone algorithm determined by turning points $\{x_i\}_{i\geq 1}$ and intermediate turning points within hops $\gamma_1,\dots, \gamma_t$. The figure also depicts all possible intervals $A_i$, $i=1,\ldots, t+1$ that the treasure can lie within a $t$-hop between $x_r$ and$x_{r+1}$. Possible placements of the treasure are depicted in every interval $A_i$, along with the first five visitations of the treasure, except the last interval $A_{t+1}$ for which there are only three visitations before the searcher returns to the origin.
}
                \label{fig: General-SubOrigin}
\label{fig:t-hop}
\end{figure}

\begin{showproof}
\begin{lemma}\label{lem:hop}
For any $j$, the time $h_j$ required for the $t$-hop $x_j \rightarrow x_{j+1}$ is
$$
h_j:=\beta^j \left( \beta+2\gamma_t-3 \right)
= \left( \beta^{j+1}-\beta^j\right) \frac{\beta+2\gamma_t-3}{\beta-1}.
$$
\end{lemma}

\begin{proof}
The reader may consult Figure~\ref{fig: General-SubOrigin}. The interval is traversed exactly three times, except from the interval $[\gamma_t\beta^r,\beta^r]$ which is traversed once. Hence, the time for a robot to move from $\beta^j$ to $\beta^{j+1}$ is
$$
3\left( \beta^{j+1}-\beta^j\right) -
2\left( \beta^{j+1}-\gamma_t \beta^j\right)
=\beta^j \left( \beta+2\gamma_t-3 \right).
$$
The alternative expression is obtained by factoring out $\left( \beta^{j+1}-\beta^j\right)$ and is given for convenience.
 \end{proof}

Using the above, we compute the total time the robot needs to progress from the origin to $\beta^r+\epsilon$.

\begin{lemma}\label{Lem:TimeToReachbr}
For any sufficiently small $\epsilon>0$, the time needed for the robot to reach $\beta^r +\epsilon$ for the first time is equal to
$$
\beta^{r}\frac{3\beta+2\gamma_t-3}{\beta-1}
-\frac{2 \beta \gamma_t}{\beta-1}
+\epsilon.
$$
\end{lemma}

\begin{proof}
The algorithm will perform a number of hops before returning to the origin after each hop. According to Lemma~\ref{lem:hop}, the total time for this trajectory is
\begin{align*}
3\beta + \sum_{j=1}^{r-1}h_j + 2\sum_{j=1}^{r-1} \beta^{j+1}
&=
3\beta + \left( \beta^{r}-\beta\right) \frac{\beta+2\gamma_t-3}{\beta-1}
+ 2\sum_{j=1}^{r-1} \beta^{j+1} \\
&=
3\beta + \left( \beta^{r}-\beta\right) \frac{\beta+2\gamma_t-3}{\beta-1}
+ 2\frac{\beta \left(\beta^r- \beta \right)}{\beta-1} \\
&=
\beta^{r}\frac{3\beta+2\gamma_t-3}{\beta-1}
-\frac{2 \beta \gamma_t}{\beta-1},
\end{align*}
and the proof follows.
\end{proof}
\end{showproof}

\subsection{Performance Analysis of $t$-Sub-Monotone Trajectories}
\label{sec: performance analysis}


For the remainder of the paper, we introduce the following expressions:
\begin{align}
A&=2(1-p)\label{equa:A},\\
B&=\frac2{\beta-1}+\frac{2(1-p)^3}{1-\beta(1-p)^2}\label{equa:B},\\
C&=\frac{2p(1-p)^3(2-p)\beta}{1-\beta(1-p)^2}\label{equa:C},\\
D&=\frac{-2 p^4+12 p^3-26 p^2+23 p-4}{2-p}\label{equa:D},\\
E&=
\frac{2p(1-p)(2-p)\beta}{1-\beta(1-p)^2}\label{equa:E},\\
F&= p \left(
2 \left(\frac{\beta(1-p)+1}{(\beta-1)(1-\beta (p-1)^2)}\right)
+ \frac{5-2 p}{2-p} \right),
\label{equa:F}
\end{align}
where, in particular,
$A=A(p),
B=B(\beta,p),
C=C(\beta,p),
D=D(p),
E=E(\beta,p),
F=F(\beta,p)
$.
The purpose of this section is to prove the following theorem.

\begin{theorem}
\label{thm: summary of worst case in intervals}
For any $i=1,\ldots,t+1$ and given that the treasure lies in interval $A_i:=(\gamma_{i-1}x_r, \gamma_{i}x_r)$, the worst case induced competitive ratio $R_i$ is given by the formula
$$
R_i =
\left\{
\begin{array}{ll}
p \left( \frac{ A\gamma_i + B \gamma_t + C}{\gamma_{i-1}} +D \right), &\mbox{if}~ i=1,\ldots,t\\
p \left(\frac{E}{\gamma_t}+F \right), &\mbox{if}~ i=t
\end{array}
\right.
$$
\end{theorem}

An immediate consequence of Theorem~\ref{thm: summary of worst case in intervals} is that the best $t$-sub-monotone algorithm with expansion factor $\beta$ within consecutive $t$-hops and intermediate turning points $\gamma_1, \gamma_2, \ldots, \gamma_t$ is the solution (if it exists) to optimization problem
\begin{align}
\min_{\beta, \gamma_1, \ldots, \gamma_t} & \max \left\{ R_1, R_2, \ldots, R_t, R_{t+1} \right\} 		\label{equa: best t-sub}\\
\textrm{s.t.} &~~~ 1<\gamma_1 <\ldots < \gamma_t <\beta < \tfrac1{(1-p)^2}.		\notag
\end{align}
Alternatively, any solution $\beta, \gamma_1, \ldots, \gamma_t$ which is feasible to~\eqref{equa: best t-sub} has competitive ratio $\max_{i=1,\ldots,t+1} R_i$.

\begin{showproof}
The proof of Theorem~\ref{thm: summary of worst case in intervals} is given by Lemmas~\ref{lem: CR cases Ai}, \ref{lem: CR cases At+1} at the end of the current section. Towards establishing the lemmas, we need to calculate the time between consecutive visitations of the treasure in order to eventually apply Lemma~\ref{lem: exp term incr} and compute the performance of a $t$-sub-monotone algorithm.

As we did previously and for the sake of simplifying the analysis, we assume that the treasure will never coincide with a turning point $\gamma_ix_j$.
Moreover, we assume that the treasure is placed at distance $d_i=\gamma_{i-1} x_r + y$ from the origin, where $0<y<(\gamma_{i}-\gamma_{i-1}) x_r$, for some $i$ that we allow for the moment to vary.

Since the treasure can be in any of these intervals, there are $t+1$ cases to consider when computing the performance of the algorithm.
Lemmas~\ref{lem: subvisitations general suborigin} and \ref{lem: subvisitations general suborigin last case} concern different cases as to where the treasure is with respect to internal turning points associated with $\gamma_i$.

\begin{lemma}\label{lem: subvisitations general suborigin}
For any $i=1,\ldots, t$, suppose that the treasure is placed at distance $d_i=\gamma_{i-1} \beta^r + y$ from the origin, where $0<y<(\gamma_{i}-\gamma_{i-1}) x_r$. We then have that
$$
g_s =
\left\{
\begin{array}{ll}
\beta^{r}\left(  \frac{2\gamma_t}{\beta-1} -\gamma_i+ 3\gamma_{i-1}  \right) -\frac{2 \beta \gamma_t}{\beta-1}  +d_i, 	&\textrm{if}~s=1 \\
2\gamma_i \beta^r - 2d_i, 	&\textrm{if}~s=2 \\
2y, 	&\textrm{if}~s=3 \\
2\beta^r\left( \beta+\gamma_t\right) -4d_i, 	&\textrm{if}~s=4 \\
2d_i, 	&\textrm{if}~s=2j+3~\textrm{for some}~j\geq 1 \\
2\beta^{r+j}\left( \beta+\gamma_t-1 \right) -2d_i, 	&\textrm{if}~s=2j+4~\textrm{for some}~j\geq 1 \\
\end{array}
\right.
$$
\end{lemma}

\begin{proof}
For computing each of the $g_j$'s we consult Figure~\ref{fig: General-SubOrigin}.
\begin{align}
g_1
&=
\beta^{r}\frac{3\beta+2\gamma_t-3}{\beta-1} -\frac{2 \beta \gamma_t}{\beta-1}
+ 3(\gamma_{i-1}-1)\beta^r +y  			\tag{By Lemma~\ref{Lem:TimeToReachbr}} \\
&=\beta^{r}\left(  \frac{3\beta+2\gamma_t-3}{\beta-1} + 3(\gamma_{i-1}-1) \right)
-\frac{2 \beta \gamma_t}{\beta-1}  +y  \notag \\
&=\beta^{r}\left(  \frac{2\gamma_t}{\beta-1} + 3\gamma_{i-1}  \right)
-\frac{2 \beta \gamma_t}{\beta-1}  +y.  \notag
\end{align}

We derive that $g_2 = 2\gamma_i \beta^r - 2d_i$, that $g_3=2y$, and that
\begin{align}
g_4
&= 4\left( \gamma_t \beta^r-d_i\right)+2\left(\beta^{r+1}-\gamma_t\beta^r\right) \notag \\
&= 2\gamma_t\beta^r + 2\beta^{r+1}-4d_i  \notag \\
&=2\beta^r\left( \beta+\gamma_t\right) -4d_i. \notag
\end{align}
After the fourth visitation of the treasure, an odd indexed visitation takes time $2d_i$; that is, $g_{2j+3}=2d_i$, for all $j\geq 1$.
Finally, for every even indexed visitation after the 4th one we have, for each $j\geq 1$, that
\begin{align}
g_{2j+4}
&=
\left( \beta^{r+j}
 -d_i \right) + h_{r+j}+\left( \beta^{r+j+1}-d_i \right)
\notag \\
&=
 \beta^{r+j}+\beta^{r+j+1}
+ \beta^{r+j} \left( \beta+2\gamma_t-3 \right)
 -2d_i
 \tag{by Lemma~\ref{lem:hop}} \\
 &= 2\beta^{r+j}\left( \beta+\gamma_t-1 \right) -2d_i, \notag
\end{align}
and the proof follows. 
 \end{proof}

\begin{lemma}\label{lem: subvisitations general suborigin last case}
Suppose that the treasure is placed at distance $d_{t+1}=\gamma_{t} \beta^r + y$ from the origin, where $0<y<(\beta-\gamma_{t}) x_r$. We then have that
$$
g_s =
\left\{
\begin{array}{ll}
\beta^{r}\left(  \frac{2\gamma_t}{\beta-1} + 3\gamma_{t}  \right) -\frac{2 \beta \gamma_t}{\beta-1}  +y, 	&\textrm{if}~s=1 \\
2 \beta^{r+1} - 2d_{t+1}, 	&\textrm{if}~s=2 \\
2d_{t+1}, 	&\textrm{if}~s=2j+1~\textrm{for some}~j\geq 1 \\
2\beta^{r+j}\left( \beta+\gamma_t-1 \right) -2d_{t+1}, 	&\textrm{if}~s=2j+2~\textrm{for some}~j\geq 1
\end{array}
\right.
$$
\end{lemma}

\begin{proof}
For the first two visitations, the time elapsed is identical to the case where the treasure is in any of the intervals $A_i$ (see Figure~\ref{fig: General-SubOrigin}). We only need to set $i=t+1$, in which case, by Lemma~\ref{lem: subvisitations general suborigin} we obtain $g_1,g_2$ as claimed (recall that $\gamma_{t+1}=\beta$). Any odd visitation thereafter will take additional time $2d_{t+1}$. Finally, every even visitation thereafter is identical to the (large indexed) even visitations of Lemma~\ref{lem: subvisitations general suborigin}, only that in the currently examined case, the index of the visitations starts from four, instead of six.
 \end{proof}

We are now ready to prove Theorem~\ref{thm: summary of worst case in intervals} by proposing and proving Lemmas~\ref{lem: CR cases Ai}, \ref{lem: CR cases At+1}, each of them describing the worst case competitive ratio over all possible placements of the treasure.

\begin{lemma}\label{lem: CR cases Ai}
For any $i=1,\ldots,t$, and given that the treasure lies in interval $A_i$, the worst case induced competitive ratio $R_i$ is given by the formula
$
R_i =p \left( \tfrac{ A\gamma_i + B \gamma_t + C}{\gamma_{i-1}} +D \right)$.
\end{lemma}

\begin{proof}
Suppose that the treasure is placed at distance $d_i=\gamma_{i-1} \beta^r + y$ from the origin, where $0<y<(\gamma_{i}-\gamma_{i-1}) x_r$. Let $C_i$ denote the expected termination time in this case. As per Lemma~\ref{lem: exp term incr}, we have that
$
C_i = \sum_{j=1}^\infty (1-p)^{i-1}g_j,
$
and recall that the competitive ratio in this case will be given by
$
p~\sup_{y,r} \frac{C_i}{d_i} = p~ \sup_{y,r} \frac{C_i}{\gamma_{i-1} x^r + y}.
$
From the above and Lemma~\ref{lem: subvisitations general suborigin} it is immediate that the largest competitive ratio is induced when $y\rightarrow 0$ (and as it will be clear momentarily, when $r\rightarrow \infty$). Therefore, in what follows we use $d_i=\gamma_{i-1}\beta^r$. We then have that

\begin{align*}
\frac{C_i}{d_i}
=&
\frac1{d_i}\left(
g_1 + (1-p)g_2 +(1-p)^3g_4 + \sum_{j\geq 1}(1-p)^{2j+2}g_{2j+3}+\sum_{j\geq 1}(1-p)^{2j+3}g_{2j+4}
\right)
\\
=&
\frac1{\gamma_{i-1}}\left(  \frac{2\gamma_t}{\beta-1} + 3\gamma_{i-1}  \right) -\frac{2 \beta \gamma_t}{\gamma_{i-1}\beta^{r}(\beta-1)}
+
2(1-p)\left( \frac{\gamma_i}{\gamma_{i-1}} - 1 \right) \\
& + 2(1-p)^3
\left(
\frac{\beta+\gamma_t}{\gamma_{i-1}} -2
\right)
+2 \sum_{j\geq 1}(1-p)^{2j+2} \\
& + \frac2{\gamma_{i-1}}\left( \beta+\gamma_t-1 \right) \sum_{j\geq 1}(1-p)^{2j+3}\beta^j
-2 \sum_{j\geq 1}(1-p)^{2j+3} \\
\stackrel{(r\rightarrow \infty)}{\leq}&
\ignore{
\frac2{\gamma_{i-1}}
\left(
 \frac{\gamma_t}{\beta-1}
 +(1-p)\gamma_i
 +(1-p)^3(\beta+\gamma_t-\gamma_i)
 +\left( \beta+\gamma_t-1 \right)\frac{\beta (1-p)^5}{1-\beta(1-p)^2}
\right) \\
&+ 3-2(1-p)-2(1-p)^3
+\frac{2 (1-p)^4}{2-p} \\
=&
}
\frac2{\gamma_{i-1}}
\left(
(1-p)\gamma_i +
+\left(\frac1{\beta-1}+\frac{(1-p)^3}{1-\beta(1-p)^2} \right)\gamma_t
+ \frac{p(1-p)^3(2-p)\beta}{1-\beta(1-p)^2}
\right) \\
&+
\frac{-2 p^4+12 p^3-26 p^2+23 p-4}{2-p},
\end{align*}
and the proof follows.
 \end{proof}

\begin{lemma}\label{lem: CR cases At+1}
Given that the treasure lies in interval $A_{t+1}$, the worst case induced competitive ratio $R_{t+1}$ is given by the formula
$R_{t+1} = p \left(\tfrac{E}{\gamma_t}+F \right)$.
\end{lemma}

\begin{proof}
We invoke  Lemma~\ref{lem: exp term incr}, which together with Lemma~\ref{lem: subvisitations general suborigin last case} allows us to compute the expected termination time $C_{t+1}$. Calculations are similar to the proof of Lemma~\ref{lem: CR cases Ai}, and in particular, the worst competitive ratio $R_{t+1}$ is induced when $y\rightarrow 0$; that is, when $d_{t+1}\rightarrow \gamma_t \beta^r$, and when $r\rightarrow \infty$. More specifically,
\begin{align*}
\sup_{r,y}
\frac{C_{t+1}}{d_{t+1}}
=& \sup_{r,y}
\frac1{d_{t+1}}\left(
g_1 + (1-p)g_2  +\sum_{j\geq 1}(1-p)^{2j}g_{2j+1}+\sum_{j\geq 1}(1-p)^{2j+1}g_{2j+2}
\right)
\\
=&
\left(  \frac{2}{\beta-1} + 3 \right)
+2 (1-p) \left( \frac{\beta}{\gamma_t}-1 \right) \\
&
+ 2 \sum_{j\geq 1}(1-p)^{2j}
+ 2\frac{ \beta+\gamma_t-1}{\gamma_{t}} \sum_{j\geq 1}(1-p)^{2j+1}\beta^j
-2 \sum_{j\geq 1}(1-p)^{2j+1} \\
=&
\frac2{\gamma_t}
\frac{p(1-p)(2-p)\beta}{1-\beta(1-p)^2}
+p \left(2 \frac{ \beta(1-p)+1}{(\beta-1)\left(1-\beta(1-p)^2\right)}+ \frac{5-2 p}{2-p} \right),
\end{align*}
and the proof follows.
 \end{proof}
\end{showproof}

\subsection{Choosing Efficient $t$-Sub-Monotone Trajectories}
\label{sec: how to choose good suborigin algo}

The purpose of this section is to propose a method for choosing parameters $\beta, \gamma_1, \ldots, \gamma_t$ of a $t$-sub-monotone algorithm which are feasible to~\eqref{equa: best t-sub}, hence, inducing competitive ratio $\max_{i=1,\ldots,t+1} R_i$. The main idea of our approach is to treat the induced competitive ratio as an unknown $R$, and then impose, for all $i=1,\ldots,{t+1}$, that $R_i=R$. The choices of $\gamma_i$ are solutions to a recurrence relation. From numerical calculations, we know that our method proposes \textit{optimal} solutions to~\eqref{equa: best t-sub}, where in particular, all strict inequality constraints are satisfied with slack. However, a proof of optimality is not evident.

For the values of $A(p),B(p,\beta),C(p,\beta),D(p), E(p,\beta), F(p,\beta)$ as defined in~\eqref{equa:A}-\eqref{equa:F}, we provide a way of obtaining $t$-sub-monotone algorithms by solving one non-linear equation. To this end, we also introduce abbreviations:
$$x:= \frac{R/p-D}{A}, ~~y:=\frac{\frac{B~E}{R/p-F} +C}{A},$$
where in particular $x=x(p,R)$ and $y=(p,\beta,R)$ (the fact that $x$ is independent of $\beta$ will be used later). Moreover, we introduce the concept of the $t$-\textit{characteristic polynomial} of a pair $(p,R)$, which is the degree-2 polynomial
$q_0+q_1\beta+q_2\beta^2$
where $q_0=q_0(p,R,t), q_1=q_1(p,R,t), q_2=q_2(p,R,t)$ are defined as
\begin{align}
q_0 =&											\label{equa: q0}
\left(p^2 (2 p ((p-6) p+12)-17)-(p-2) R\right) \left(p^2+(p-2) R\right)
x^t 						\\
q_1 =&											\label{equa: q1}
2 (p-2)^4 (p-1) p^3 (R-p)+x^t \times \\
&\left((p (p (2 p (p (2 p-19)+74)-297)+308)-134) p^4 \right. \notag \\
&~\left. -2 (p-2) (p (p ((p-8) p+25)-35)+20) p^2 R-(p-2)^2 ((p-2) p+2) R^2\right)
\notag \\
q_2 =&											\label{equa: q2}
(p-1) \left(2 (p-2)^4 p^3 (3 p-R)\right. \\
&~~~~~~~~~~~~\left. - 	(p-1)
			\left(p^2 (2 p-5)-(p-2) R\right)
			\left((2 (p-4) p+9) p^2+(p-2) R\right)
			x^t
		\right) \notag
\end{align}
Note that the discriminant of the $t$-characteristic polynomial of a pair $(p,R)$ is a rational function of $p,R$ (where the numerator and denominator are polynomials of degree $\Theta(t)$), and hence, a function exclusively of $R$, for every fixed $p$.

Given $p\in (0,1)$, we say that pair $(\beta,R)$ is \textit{feasible} if
\begin{align}
& x-y-1 >0,							\label{equa: con1}\\
& \beta - \frac{E}{R/p-F} >0.		\label{equa: con2}
\end{align}
As we shall see, constraints above guarantee that $\beta$ is a valid expansion factor, and that the last turning point of a sub-monotone algorithm happens before a $t$-hop is completed.
We will also require that 
\begin{equation}
\label{equa: nonlinearequality}
\left( 1-\frac{y}{x-1}\right) x^t + \frac{y}{x-1} - \frac{E}{R/p-F}=0.
\end{equation}
As the treasure could be located in any of the $t+1$ sub-intervals associated with a $t$-hop, constraint~\eqref{equa: nonlinearequality} will guarantee that the competitive ratio is independent of that placement. 
Our main theorem is the following.
\begin{theorem}
\label{thm: best t-sub-origin}
Fix $p \in (0,1)$, and let $R\geq 3$ be such that the discriminant of the $t$-characteristic polynomial of pair $(p,R)$ is equal to 0. Let $\beta=-q_1/2q_2$ and suppose that pair $(\beta,R)$ is feasible. We also set
$
\gamma_i = \left( 1-\frac{y}{x-1}\right) x^i + \frac{y}{x-1},~i=1\ldots, t
$.
We then have that $\beta, \gamma_1, \ldots, \gamma_t$ is a $t$-sub-monotone algorithm with competitive ratio $R$ for problem \FS.
\end{theorem}

\begin{showproof}
The main ingredient for proving Theorem~\ref{thm: best t-sub-origin} is the following lemma.

\begin{lemma}
\label{lem: best t-sub-origin}
For some $p \in (0,1)$, consider values of $t,R,\beta$ 
satisfying constraint~\eqref{equa: nonlinearequality}. 
\ignore{
\begin{equation}
\label{equa: nonlinearequality}
\left( 1-\frac{y}{x-1}\right) x^t + \frac{y}{x-1} - \frac{E}{R/p-F}=0.
\end{equation}
}
If additionally, the pair $(\beta,R)$ is feasible, then $R$ is the competitive ratio of a $t$-sub-monotone trajectory with parameters $\beta, \gamma_1, \ldots, \gamma_t$ for problem \FS, where
$
\gamma_i = \left( 1-\frac{y}{x-1}\right) x^i + \frac{y}{x-1},~i=1\ldots, t
$.
\end{lemma}

\ignore{
The solution to the optimization problem is then
\begin{align*}
\inf_{R,\beta, \gamma_i} ~~~& R \\
\textrm{s.t.}~~ &
\left( 1-\frac{y}{x-1}\right) x^t + \frac{y}{x-1} = \frac{E}{R/p-F} \\
& x-y >1 \\
& \frac{E}{R/p-F} < \beta,
\end{align*}
gives the competitive ration $R$ of a $t$-sub-monotone trajectory with parameters $\beta, \gamma_1, \ldots, \gamma_t$.
}

\begin{proof}
By Theorem~\ref{thm: summary of worst case in intervals}, the best $t$-sub-monotone algorithm is determined by parameters $\gamma_1, \gamma_2, \ldots, \gamma_t, \beta$ that minimize $\max \left\{ R_1, R_2, \ldots, R_t, R_{t+1} \right\}$, subject to that $1<\gamma_1 <\ldots < \gamma_t <\beta < \frac1{(1-p)^2}$.
The bound on $\beta$ guarantees convergence of the expected termination time. We attempt to find a solution to the optimization problem above by requiring that
$$
R_1 = R_2 = \ldots = R_t = R_{t+1}.
$$
Denote the value of the optimal solution by $R$, and suppose that it is realized by parameters $\gamma_1, \gamma_2, \ldots, \gamma_t, \beta$.
By Lemma~\ref{lem: CR cases At+1}, we have that
\begin{equation}\label{equa: opt gt}
\gamma_t = \frac{E}{R/p-F}
\end{equation}
We then have that by Lemma~\ref{lem: CR cases Ai} and solving for $\gamma_i$ we obtain that for each $i=1,\ldots, t$
\begin{align*}
\gamma_i &=\frac{R/p-D}{A}\gamma_{i-1} - \frac{B\gamma_t +C}{A}\\
			&\stackrel{\eqref{equa: opt gt}}{=} \frac{R/p-D}{A}\gamma_{i-1} - \frac{\frac{B~E}{R/p-F} +C}{A},
\end{align*}
with the understanding that $\gamma_0=1$.
Hence, the recurrence relation for $\gamma_i$ gives
$$
\gamma_i = \left( 1-\frac{y}{x-1}\right) x^i + \frac{y}{x-1}, ~~i=1\ldots, t.
$$
The last expression for $\gamma_i$, when $i=t$ should agree with~\eqref{equa: opt gt}.
It is straightforward to see that since $R\geq 3$, we obtain that $x>1>0$. So condition $\gamma_i >\gamma_{i-1}$ translates into that $x-y>1$, which also guarantees that $\gamma_1>1$. Finally, the last condition asserts that $\gamma_t < \beta$.
 \end{proof}

Theorem~\ref{thm: best t-sub-origin} suggests that in order to obtain an efficient $t$-sub-monotone algorithm with parameters $\beta, \gamma_1, \ldots, \gamma_t$, we need to minimize $R$ subject to constraint~\eqref{equa: nonlinearequality} (and to the associated strict inequality constraints). Ideally, we would like to find all roots to the associated (at least) degree-$t$ polynomial in $R$, and identify the minimum root that complies with the remaining feasibility conditions. The task is particularly challenging (from a numerical perspective), since that polynomial's coefficients depend also on the unknown value $\beta$. To bypass this difficulty, and for fixed $p,t$, we define intuitive values of $R,\beta$ that always satisfy the constraint, for which we need to check separately that they induce valid search trajectories (which is established by checking the two strict inequalities). Numerical calculations suggest that this heuristic choice of $R,\beta$ is the optimal one, but a proof is eluding us. Nevertheless, the choice of $R,\beta$ is valid, which is summarized by the statement of Theorem~\ref{thm: best t-sub-origin} and which we are ready to prove next.

\medskip

\noindent \emph{Proof of Theorem}~\ref{thm: best t-sub-origin}.
Expression~\eqref{equa: nonlinearequality} is a rational function on $\beta$. Tedious (and software assisted symbolic calculations) show that the numerator of that rational function is the $t$-characteristic polynomial $q_0+q_1\beta+q_2\beta^2$ of pair $(p,R)$.
If $R$ is such that the discriminant of that polynomial is equal to 0, then $-q_1/2q_2$ is a root to the polynomial, and hence, constraint~\eqref{equa: nonlinearequality} is satisfied for the values of $p,R,\beta,t$.
Since pair $(\beta,R)$ is feasible,
all preconditions of Lemma~\ref{lem: best t-sub-origin} are satisfied, and hence, $\beta,\gamma_1, \ldots, \gamma_t$ is a $t$-sub-monotone algorithm with competitive ratio $R$ for problem \FS. 
\medskip

We observe that Theorem~\ref{thm: best t-sub-origin} computes exactly the best monotone algorithm of Lemma~\ref{lem: upper bound backtoorigin}.  In other words, the $0$-sub-monotone we propose above is the optimal monotone algorithm we have already studied. Indeed, the discriminant of the $0$-characteristic polynomial of $(p,R)$ equals
$$
(p-2)^2 p^2 (p (p (17-2 p ((p-6) p+12))+R)-2 R)^2 \left((p+R) \left((p-2)^2 R+p ((p-4) p+12)\right)-16 R\right)
$$
The two roots of the right-hand-side factor above is a degree-2 polynomial in $R$ with roots $\frac{4\pm4\sqrt{1-p}}{2-p}-p$, one of which (the only one which is at least $3$) being exactly the competitive ratio calculated by Lemma~\ref{lem: upper bound backtoorigin}. Moreover, setting $\beta=-q_1/2q_2$ gives the same value of the expansion factor, which is denoted by $b$ in Lemma~\ref{lem: upper bound backtoorigin}.

\end{showproof}

\subsection{Numerical Computation of $t$-Sub-Monotone Trajectories, $t\leq 10$}
\label{sec: t suborigin algo, t<=10}

We summarize the numerical results we obtain by invoking Theorem~\ref{thm: best t-sub-origin} for $t=1, \ldots, 10$, obtaining $t$-sub-monotone algorithms that induce better and better competitive ratios. For each $t$ and (enough many) $p \in (0,1)$ we compute the smallest root $R=R(p,t)$ at least 3 of the $t$-characteristic polynomial, and the associated value of the expansion factor $\beta=\beta(p,t)$. For every pair $(\beta,R)$ we verify that the induced values of $\gamma_i$ do define a feasible search trajectory by showing that pair $(\beta,R)$ is feasible. Note that constraints~\eqref{equa: con1} and ~\eqref{equa: con2} guarantee that $\beta$ is a valid expansion factor, and that the intermediate turning points of a $t$-hop are well defined, assuming that the worst case competitive ratio is the same in all subintervals of a $t$-hop, as required by constraint~\eqref{equa: nonlinearequality}.

The improvement in the competitive ratio, when $t=1,\ldots,4$ is apparent from a plot of the competitive ratio as a function of $p$, see Figure~\ref{fig: CompetitiveRationtSuborigin}. Figure~\ref{fig: ValuesofBetaUpto4} displays the behavior of the expansion factors $\beta$. Finally, Figures~\ref{fig: Constraint1Upto4} and~\ref{fig: Constraint2Upto4} confirm that the proposed solution is valid (by checking constraints~\eqref{equa: con1} and ~\eqref{equa: con2}), or in other words that the reported competitive ratio of Figure~\ref{fig: CompetitiveRationtSuborigin} is correct. The horizontal axis in all figures is probability $p$. The vertical axis is explained in detail in each of the captions.

\begin{minipage}{\linewidth}
      \centering
      \begin{minipage}{0.45\linewidth}
          \begin{figure}[H]
              \includegraphics[width=2in]{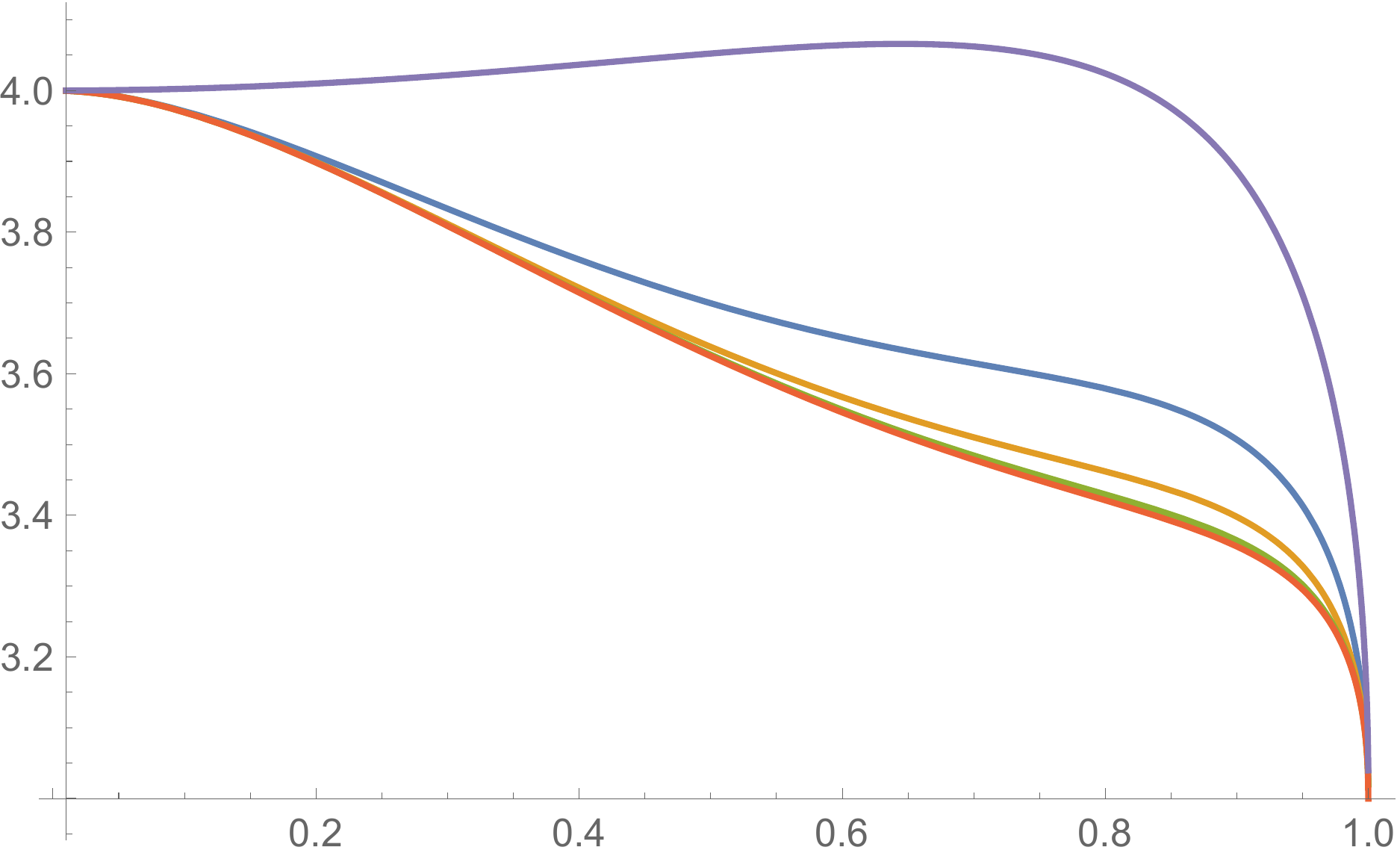}
              \caption{The vertical axis shows the behavior of the achieved competitive ratio $R_t=R_t(p)$ of various $t$-sub-monotone algorithms. Purple line corresponds to the monotone algorithm of Lemma~\ref{lem: upper bound backtoorigin}; that is, when $t=0$.  The subsequent improvements for $t=1,2,3,4$ are shown in colors blue, yellow, green and red, respectively.}
              \label{fig: CompetitiveRationtSuborigin}
          \end{figure}
      \end{minipage}
      \hspace{0.05\linewidth}
      \begin{minipage}{0.45\linewidth}
          \begin{figure}[H]
              \includegraphics[width=2in]{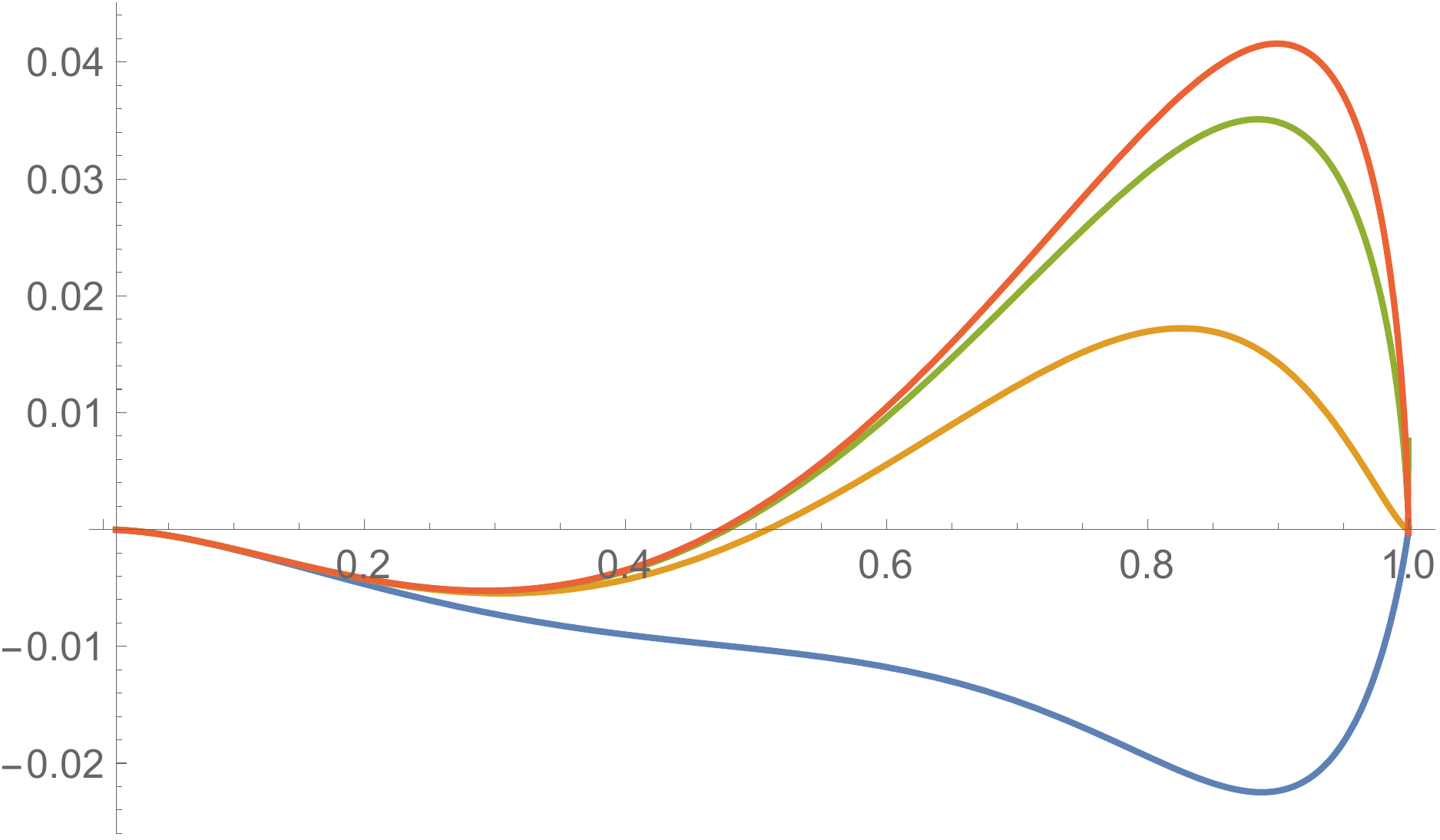}
              \caption{Figure depicts the behavior of the proposed expansion factors $\beta_t=\beta_t(p)$ for various values of $t$, as a function of $p\in (0,1)$ (horizontal axis), that induce the competitive ratios $R_t$ depicted in Figure~\ref{fig: CompetitiveRationtSuborigin}.
For the sake of better comparison, the vertical axis is $\beta_t (1-p)^2 - (1-p)$, which also shows that each expansion factor is more than 1 and less that $1/(1-p)^2$. Colors blue, yellow, green and red correspond to $t$-sub-monotone algorithms $t=1,2,3$ and $4$, respectively.      }
              \label{fig: ValuesofBetaUpto4}
          \end{figure}
      \end{minipage}
  \end{minipage}

\begin{minipage}{\linewidth}
      \centering
      \begin{minipage}{0.45\linewidth}
          \begin{figure}[H]
              \includegraphics[width=2.in]{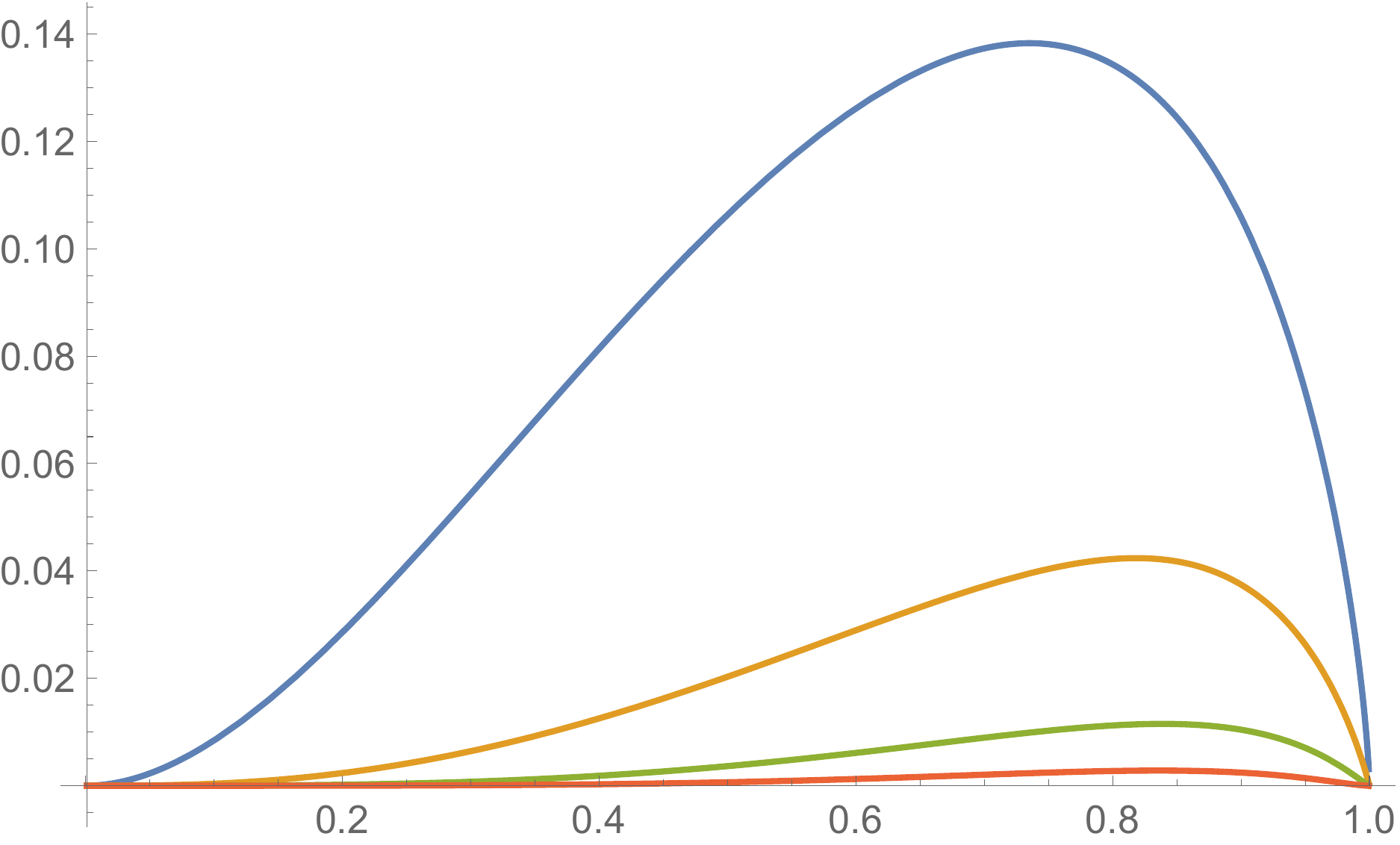}
              \caption{
              This figure shows why the choices of $\beta_t(p)$ of Figure~\ref{fig: ValuesofBetaUpto4}, that induce the competitive ratios $R_t(p)$ of Figure~\ref{fig: CompetitiveRationtSuborigin}, satisfy constraint~\eqref{equa: con1} as required by Theorem~\ref{thm: best t-sub-origin}.
Recall that $x,y$ are functions of $p,R_t,\beta_t$.
For the sake of better comparison, the vertical axis corresponds to $(x-y-1)p(1-p)$.
Colors blue, yellow, green and red correspond to $t$-sub-monotone algorithms $t=1,2,3$ and $4$, respectively.
}
              \label{fig: Constraint1Upto4}
          \end{figure}
      \end{minipage}
      \hspace{0.05\linewidth}
      \begin{minipage}{0.45\linewidth}
          \begin{figure}[H]
              \includegraphics[width=2.in]{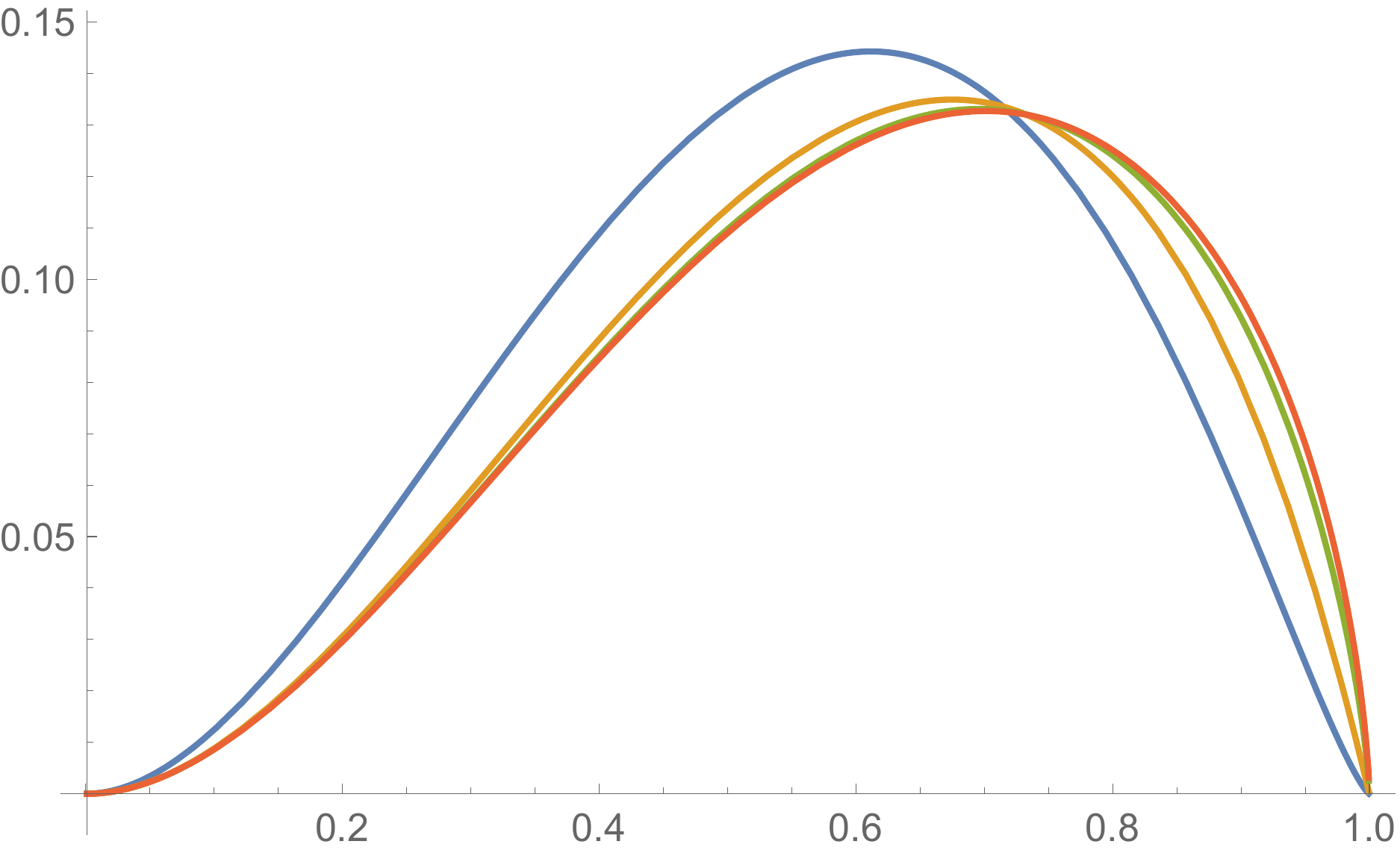}
              \caption{
              This figure shows why the choices of $\beta_t(p)$ of Figure~\ref{fig: ValuesofBetaUpto4}, that induce the competitive ratios $R_t(p)$ of Figure~\ref{fig: CompetitiveRationtSuborigin}, satisfy constraint~\eqref{equa: con2}, as required by Theorem~\ref{thm: best t-sub-origin}. Recall that $E,F$ are functions of $p,R_t,\beta_t$.
For the sake of better comparison, the vertical axis corresponds to $\left(\beta_t-\frac{E}{R_t/p-F}\right)(1-p)^2$.
Colors blue, yellow, green and red correspond to $t$-sub-monotone algorithms $t=1,2,3$ and $4$, respectively.
}
              \label{fig: Constraint2Upto4}
          \end{figure}
      \end{minipage}
  \end{minipage}
  ~\\

For values $t=5,\ldots,10$ we need to deploy heuristic comparisons in order to display the behavior of the achieved competitive ratio, along with the corresponding expansion factor (this is due to that improvements are negligible, even though strictly positive). Figure~\ref{fig: CompRatio, beta and constraint t=5-10}-left compares the achieved competitive ratios.
Figure~\ref{fig: CompRatio, beta and constraint t=5-10}-middle displays the relative behavior of the expansion factors. Finally, Figure~\ref{fig: CompRatio, beta and constraint t=5-10}-right shows why the proposed solution satisfies constraint~\eqref{equa: con1} of Theorem~\ref{thm: best t-sub-origin}. As for constraint~\eqref{equa: con2}, numerical calculations suggest that expression $ \beta - \frac{E}{R/p-F}$ remains nearly invariant for $t\geq 5$, and hence, showing the behavior for $t=5,\ldots,10$ results in a degenerate figure where all curves nearly coincide (see also Figure~\ref{fig: Constraint2Upto4}, where expressions for $t=3,4$, green and red respectively, are already very close to each other).
The horizontal axis is always probability $p\in (0,1)$, while different $t$-sub-monotone algorithms are displayed with different colors. The values of the vertical axes are described in the corresponding captions.

\def\length{0.3}
 \begin{figure}[h!]
\centering
  \includegraphics[width=\length\linewidth]{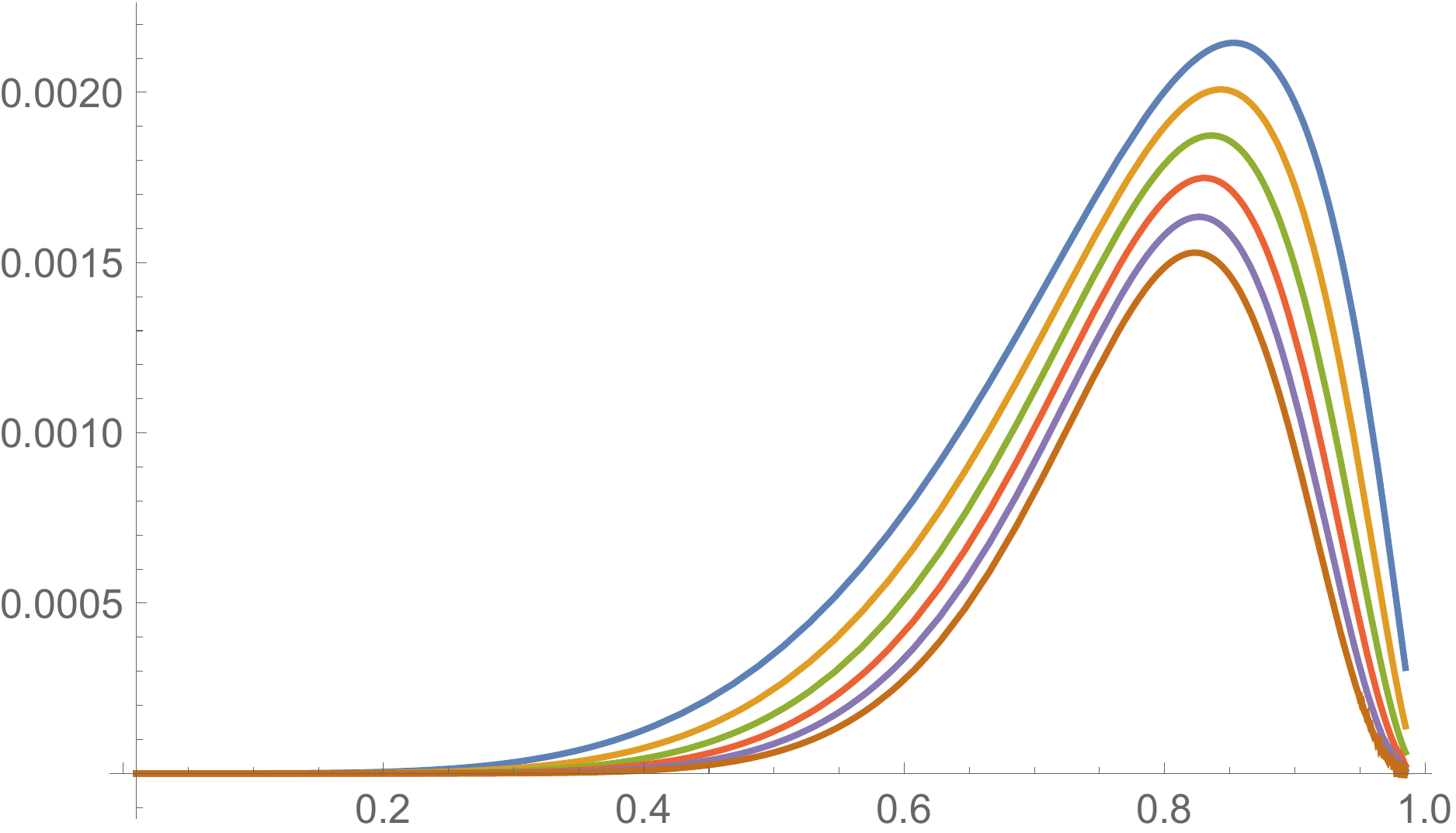}
~~~ \includegraphics[width=\length\linewidth]{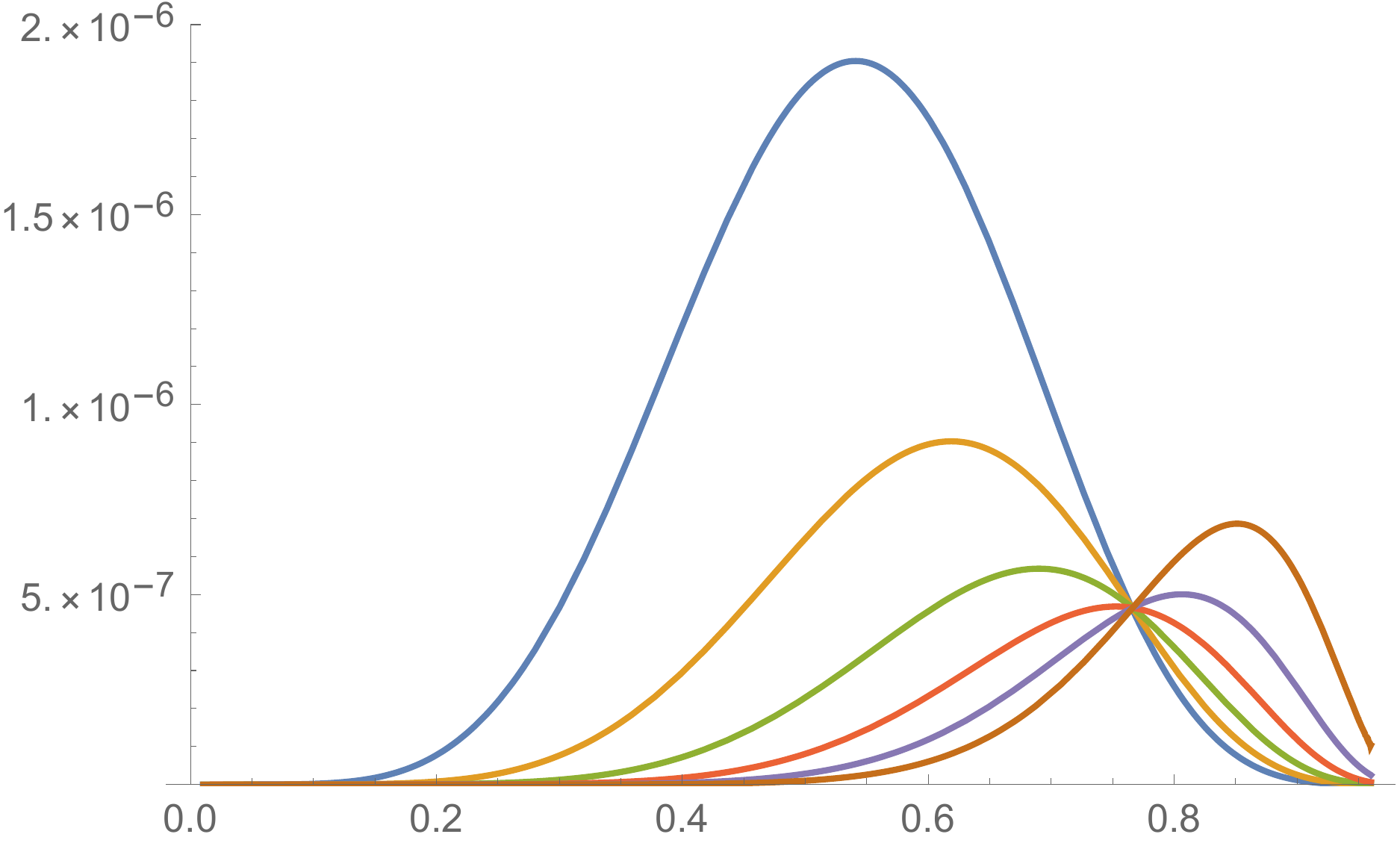}
~~~~  \includegraphics[width=\length\linewidth]{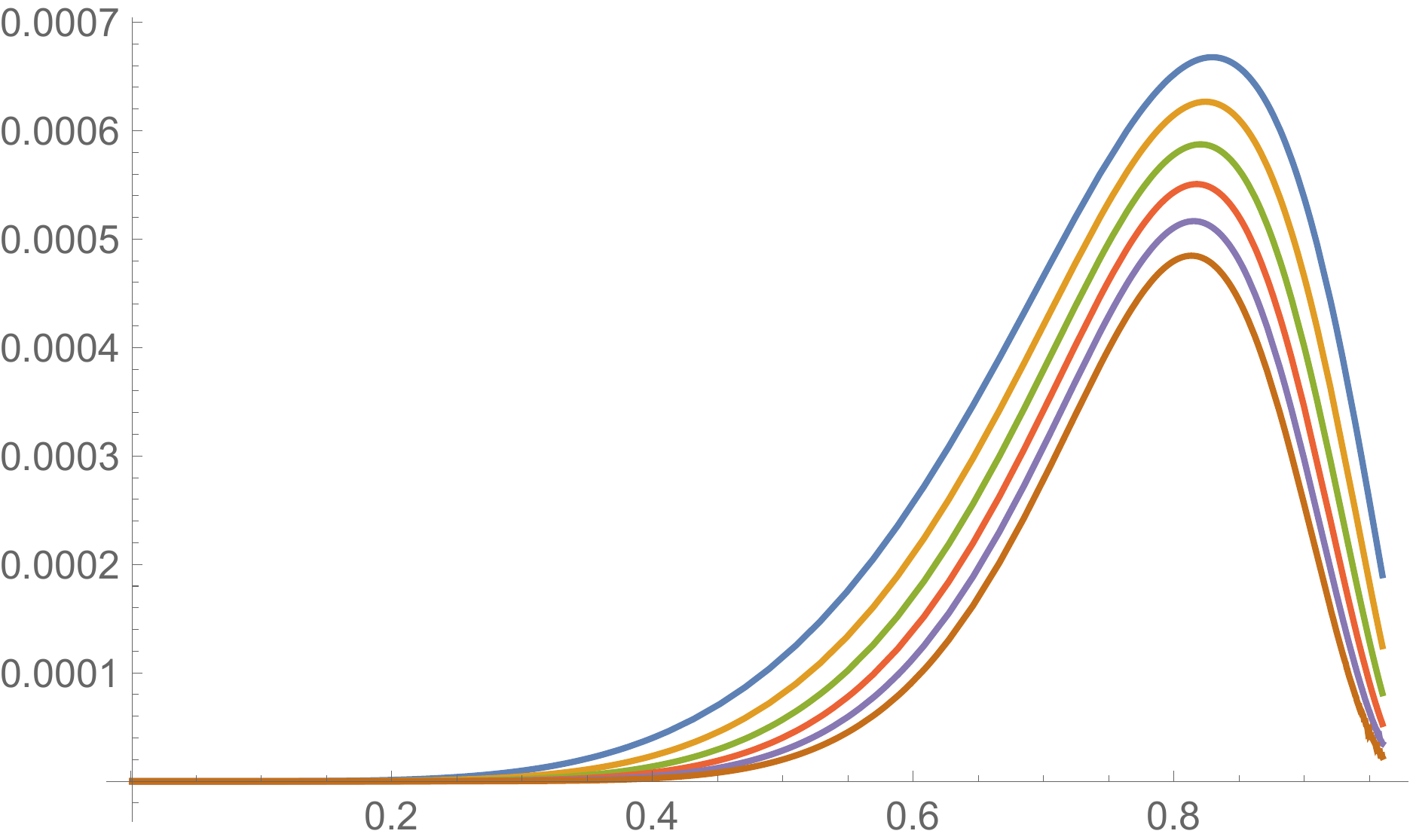}
\caption{
The figures summarize the behavior of $t$-sub-monotone algorithms for $t=5,6,7,8,9,10$, see colors blue, yellow, green, red, purple and brown, respectively. All the horizontal axes correspond to $p\in (0,1)$.  
- Left figure displays the behavior of the achieved competitive ratio $R_t$.
For each $t=5,\ldots,10$, the vertical axis corresponds to the scaled marginal improvements $4^{t-5}(R_{t-1}-R_{t})$ between two consecutive values of $t$, which show that the competitive ratio does improve with $t$, still the improvement is increasingly negligible. The scalar was introduced so that the competitive ratios can be displayed together. 
- Middle figure displays the behavior of the expansion factors $\beta_t$ that give rise to competitive ratios $R_t$. For each $t=5,\ldots,10$, the vertical axis is the scaled relative change $(1-p)^{11-t}(\beta_t-\beta_{t-1})/\beta_t$, where the scalar was introduced to improve comparison. 
- Right figure shows that the values of $\beta_t, R_t$ chosen, as per the left and middle figures, do indeed satisfy constraint~\eqref{equa: con1} of Theorem~\ref{thm: best t-sub-origin}. The horizontal axis corresponds to $(x-y-1)p(1-p)4^{t-5}$, where the scalars were introduced so that plots are comparable.
}
\label{fig: CompRatio, beta and constraint t=5-10}
\end{figure}

\subsection{Some Closed Formulae \& the Case $t\rightarrow\infty$}
\label{sec: heuristics}

As already discussed, we conjecture that the $t$-sub-monotone algorithms derived by Theorem~\ref{thm: best t-sub-origin} are optimal solutions to optimization problem~\eqref{equa: best t-sub}, even though our conjecture does not compromise the correctness of our algorithms for problem \FS. Nevertheless, a disadvantage of our approach, and in general of $t$-sub-monotone algorithms, is that our choices of parameters $\beta, \gamma_1, \ldots, \gamma_t$ do not admit closed form descriptions as functions of $p$.
In this section, we deviate from our goal to determine the best possible $t$-sub-monotone algorithms, and we present specific choices of parameters $\beta, \gamma_1, \ldots, \gamma_t$ with closed formulas which induce nearly optimal competitive ratios.

Apart from our monotone trajectories, all our positive results were summarized in Section~\ref{sec: t suborigin algo, t<=10} and were based on numerical, and computer assisted, calculations. In light of Theorem~\ref{thm: best t-sub-origin}, it is immediate that closed formulas for the achieved competitive ratios of $t$-sub-monotone algorithms do not exist. An exception, apart from the degenerate case $t=0$, is the case $t=1$. In particular, the discriminant of the $1$-characteristic polynomial of pair $(p,R)$ can be factored in two polynomials in $R$ of degree 4 and of degree 2. One of the roots to the degree-4 polynomial is the competitive ratio of the $1$-sub-monotone algorithm (as also per Theorem~\ref{thm: best t-sub-origin}). Hence, the achieved competitive ratio $R$ of the $1$-sub-monotone algorithm, along with the corresponding expansion factor $\beta$ (depicted in Figures~\ref{fig: CompetitiveRationtSuborigin},\ref{fig: ValuesofBetaUpto4}, respectively) admit closed formulas, even though they are enormous. Nevertheless, we show in the next theorem how to obtain an $1$-sub-monotone and nearly optimal algorithm with performance and expansion factor that admit elegant closed formulas (see Figure~\ref{fig: heuristic t=1,2}-left for comparison to the $1$-sub-monotone algorithm of Theorem~\ref{thm: best t-sub-origin}).
Note that Theorem~\ref{thm: back to origin optimal} combined with Theorem~\ref{thm: heuristic t=1} below show provably, and not (computer-assisted and) numerically, that monotone algorithms are strictly sub-optimal for \FS, for all $p \in (0,1)$.

 \begin{figure}[h!]
\centering
  \includegraphics[width=\length\linewidth]{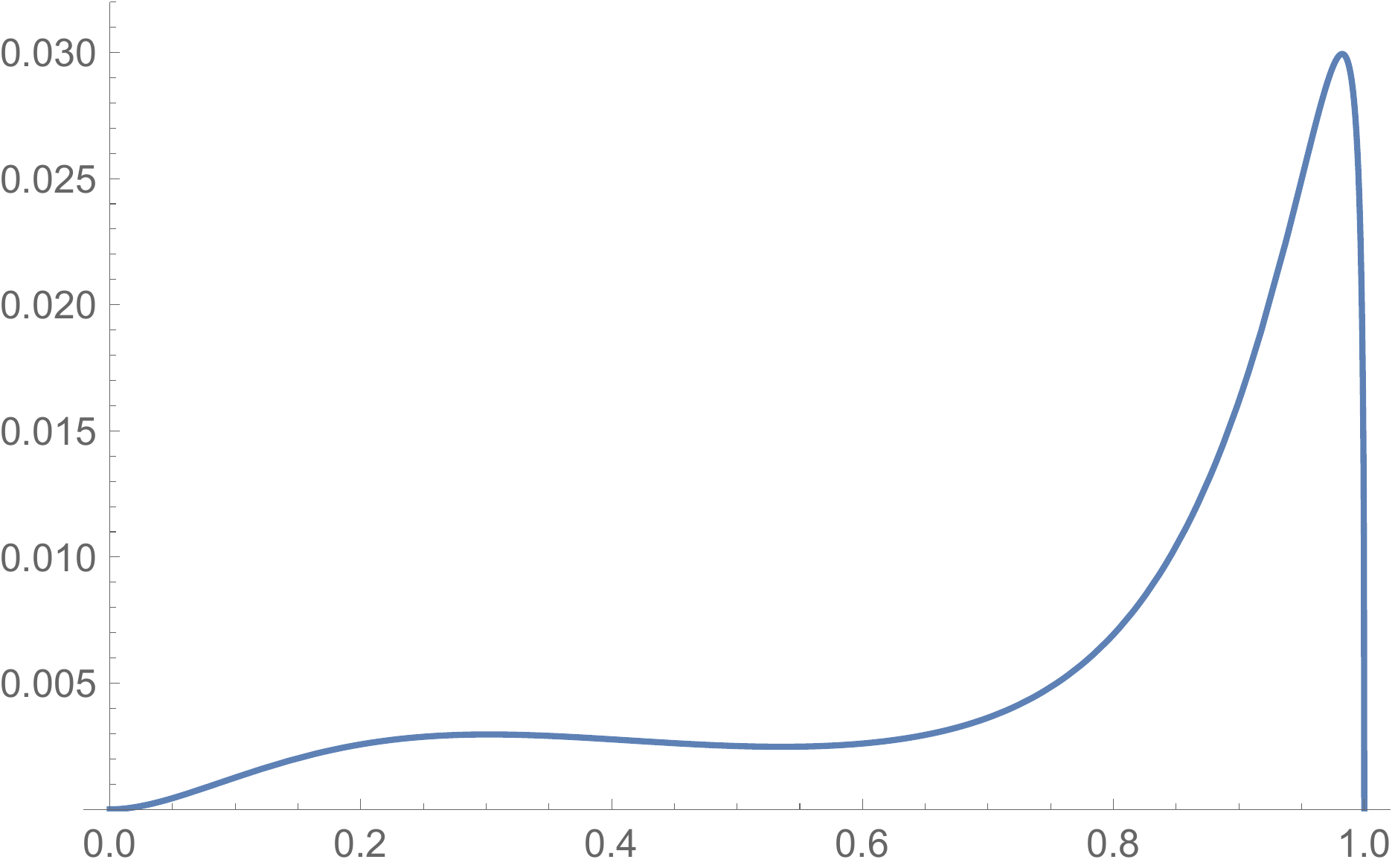}
~~~  \includegraphics[width=\length\linewidth]{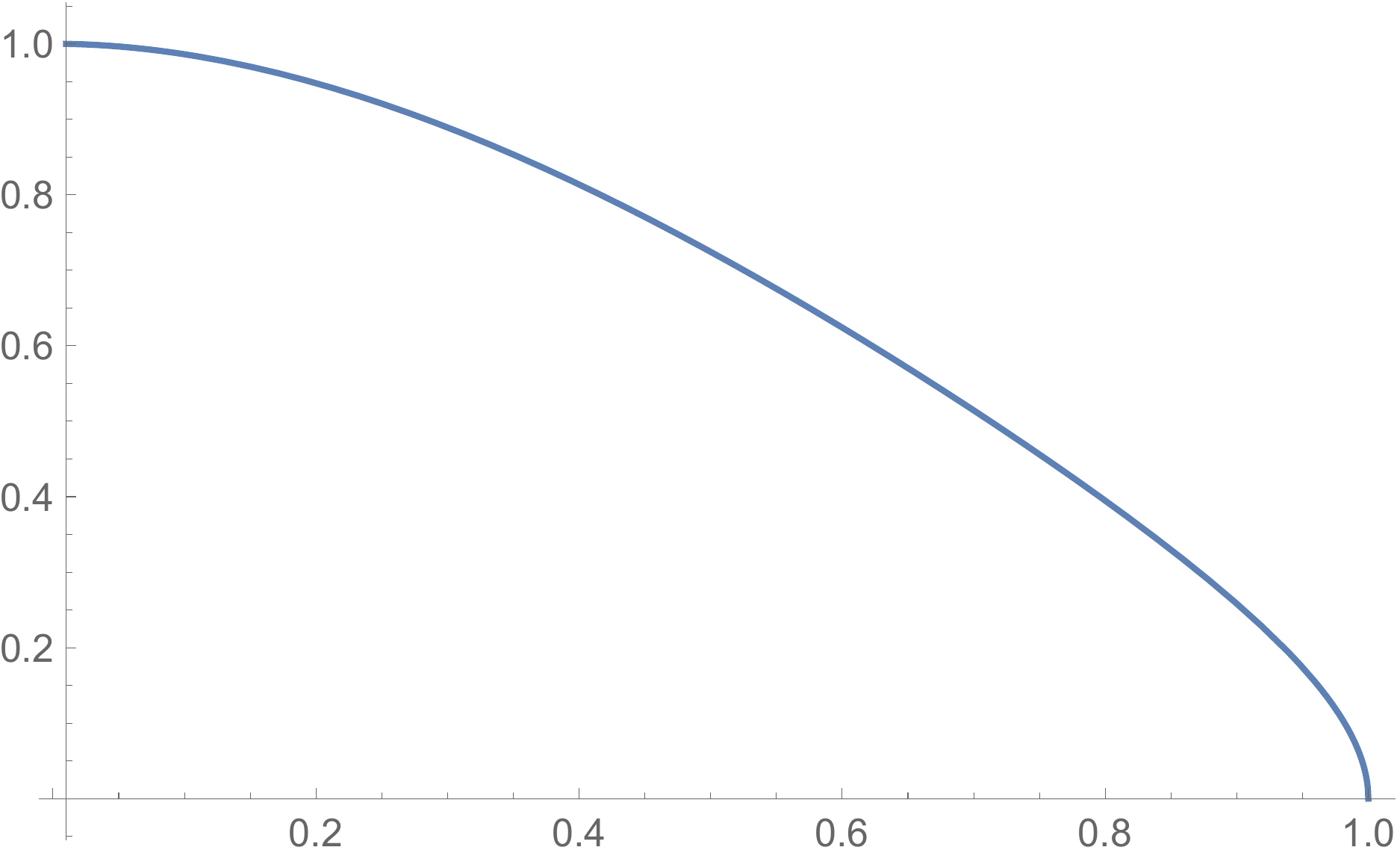}
~~~  \includegraphics[width=\length\linewidth]{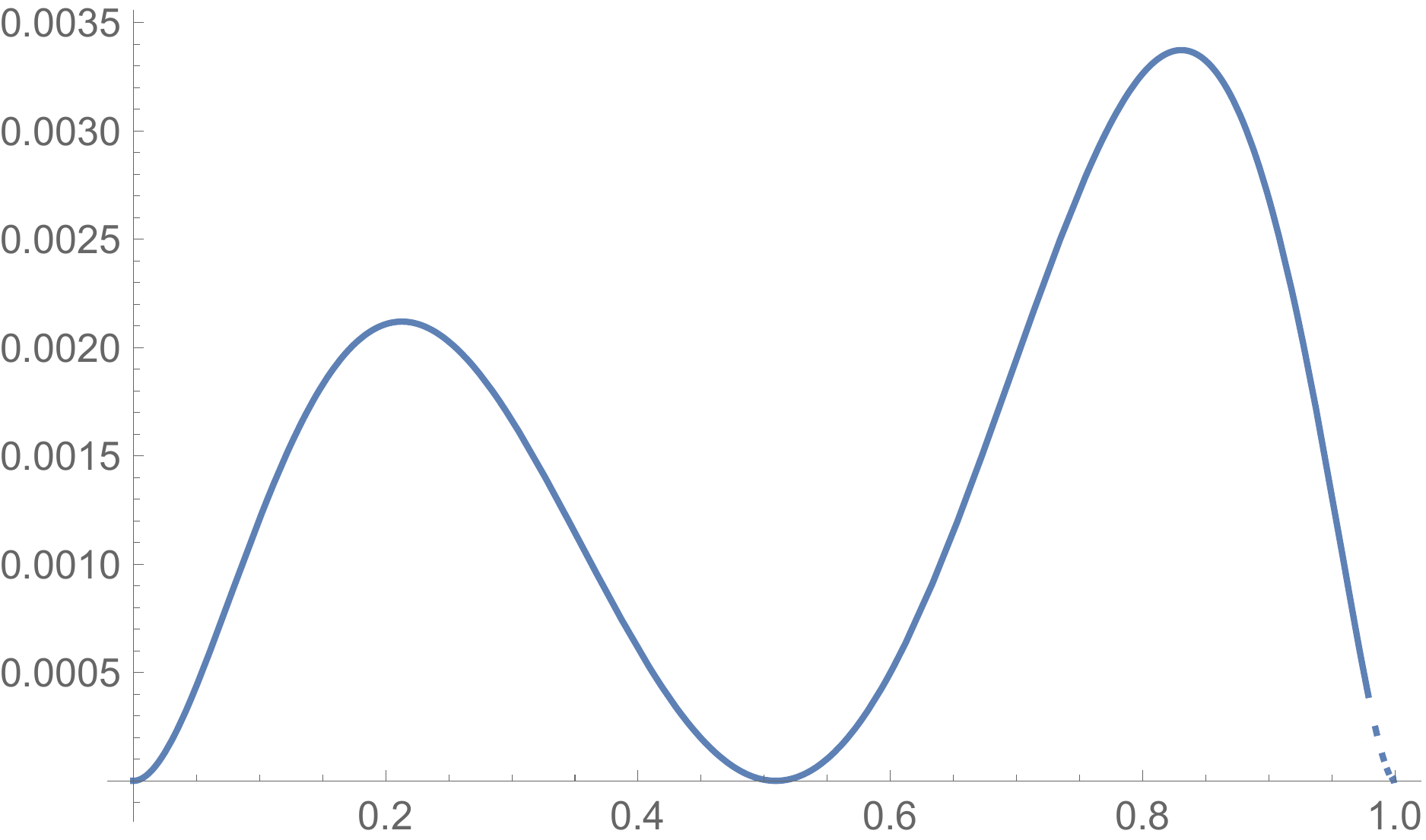}
\caption{
Figures depict the performance of $t$-sub-monotone algorithms ($t=1,2$) with expansion factors $\beta=1/(1-p)$, as a function of $p\in (0,1)$. 
- Left Figure shows the difference between the competitive ratio achieved by Theorem~\ref{thm: heuristic t=1} and the competitive ratio of the $1$-sub-monotone algorithm induced by Theorem~\ref{thm: best t-sub-origin}.  
- Middle Figure shows the behavior of the intermediate turning point $\gamma_1$ of the 1-Hop of $1$-sub-monotone algorithm, compared to expansion factor $\beta$. The vertical axis equals $\gamma_1/\beta=\gamma_1(1-p)$, which is shown to be at most 1, as wanted. 
- Right Figure
shows the difference between the competitive ratio achieved by a heuristic $2$-sub-monotone algorithm using $\beta=1/(1-p)$ and the competitive ratio of the $1$-sub-monotone algorithm induced by Theorem~\ref{thm: best t-sub-origin}.
}
\label{fig: heuristic t=1,2}
\end{figure}

\begin{theorem}
\label{thm: heuristic t=1}
There is a $1$-sub-monotone algorithm for \FS\
with competitive ratio
$$R=\sqrt{(p-2) (p-1) (p (p (4 p-3)+5)+2)}+\frac{4}{2-p}-(2-p) p,$$
and expansion factor $\beta=1/(1-p)$.
\end{theorem}

\begin{showproof}
\begin{proof}
We fix $\beta=1/(1-p)$ and invoke constraint~\eqref{equa: nonlinearequality}, so as to force that the competitive ratio does not depend on which subinterval the treasure is placed within a 1-Hop of the $1$-sub-monotone-algorithm. The constraint then becomes
$$
\frac{1}{2} \left(\frac{R}{p-p^2}+\frac{p-4}{p^2-3 p+2}-\frac{4 ((p-1) p+2) (p-2)^2}{p (p (2 p-9)-R+12)+2 (R-4)}\right)=0,
$$
which, solved for $R$, gives the promised competitive ratio.

As for the turning point $\gamma_1$ of the $1$-Hop, it can be computed as $\frac{E}{R/p-F}$ and in order to be valid, it has to be positive and at most $\beta=1/(1-p)$. This is verified in Figure~\ref{fig: heuristic t=1,2}-middle.
 \end{proof}
\end{showproof}

Similar to Theorem~\ref{thm: heuristic t=1}, it is possible to identify a $2$-sub-monotone algorithm with nearly optimal solution. Indeed, choosing again $\beta=1/(1-p)$ and for $t=2$, constraint~\eqref{equa: nonlinearequality} becomes
\begin{align*}
& \frac{1}{4} \left(\frac{R^2}{(p-1)^2 p^2}-\frac{8 \left((p ((p-5) p+10)-7) p^2+4\right) (p-2)^2}{p (p (p (2 p-9)-R+12)+2 (R-4))}\right. \\
& ~~~~\left. -4 p^2+\frac{p (p (-4 (p-7) p-71)+72)-16}{\left(p^2-3 p+2\right)^2}-\frac{2 (p ((p-6) p+13)-11) R}{(p-2) (p-1)^2}+10 p-\frac{16}{p}\right)=0,
\end{align*}
which can be converted into a degree-3 polynomial equation in $R$. The real root of that polynomial is the competitive ratio of a 2-sub-monotone algorithm, whose performance compared to the competitive ratio induced by Theorem~\ref{thm: best t-sub-origin} is shown in Figure~\ref{fig: heuristic t=1,2}-right.

We now turn our attention to the best competitive ratio we can achieve by $t$-sub-monotone algorithms if we allow $t$ to grow. By Section~\ref{sec: t suborigin algo, t<=10}, and in particular Figure~\ref{fig: CompRatio, beta and constraint t=5-10}, we know that the additive improvement in the competitive ratio, at least when $t\leq 10$, reduces almost by a factor of 4 between consecutive values of $t$. Interestingly, we can determine the limit $R_t$ as $t\rightarrow \infty$. The key observation is that if for some $p,R$ we have that $x(p,R)$ is bounded away from 1, then $x^t$ would be dominant in constraint~\eqref{equa: nonlinearequality}. Equivalently, the $t$-characteristic polynomial of pair $(p,R)$ (see also~\eqref{equa: q0}, \eqref{equa: q1}, \eqref{equa: q2}) would converge, as $t$ grows, to the polynomial $x^t \left( \bar q_0 + \bar q_1 \beta + \bar q_2 \beta^2\right) $, where
\begin{align*}
\bar q_0 =&		
\left(p^2 (2 p ((p-6) p+12)-17)-(p-2) R\right) \left(p^2+(p-2) R\right)
					\\
\bar q_1 =&											
\left((p (p (2 p (p (2 p-19)+74)-297)+308)-134) p^4 \right.  \\
&~\left. -2 (p-2) (p (p ((p-8) p+25)-35)+20) p^2 R-(p-2)^2 ((p-2) p+2) R^2\right) \\
\bar q_2 =&											
 - 	(p-1)^2
			\left(p^2 (2 p-5)-(p-2) R\right)
			\left((2 (p-4) p+9) p^2+(p-2) R\right)
\end{align*}
The discriminant of the polynomial would then become $\bar q_1^2-4\bar q_0 \bar q_2$ which is a degree 4 polynomial in $R$. Therefore, its four roots can be computed by closed formulas. Numerical calculations show that the polynomial in $R$ has two imaginary roots (for every $p\in (0,1)$), one real root less than 1 and one root at least 3, which we denote by $\bar R=\bar R(p)$. By Theorem~\ref{thm: best t-sub-origin}, $\bar R$ would be the limit of the competitive ratios achieved by $t$-sub-monotone algorithms, assuming that the sequence of $(\beta_t, R_t )$ is feasible.

In Figure~\ref{fig: Result asymptotic t}-left we compare $\bar R$ against the $10$-sub-monotone algorithm we established before, showing this way that the improvement we can achieve against monotone algorithms is well illustrated in Figure~\ref{fig: CompetitiveRationtSuborigin}. Indeed, Figure~\ref{fig: Result asymptotic t}-left shows that already when $t=10$ the achieved competitive ratio is within less that $10^{-6}$ additively off from the best competitive ratio we can achieve if we let $t$ grow, for every $p\in (0,1)$.
Next, Figure~\ref{fig: Result asymptotic t}-middle shows that $x(p,\bar R)$ is bounded away from 1 for all $p\in (0,1)$ and for the computed value $\bar R$, therefore, $\bar R$ is  the limit of values $R_t$ that makes the discriminant of the $t$-characteristic polynomial equal to 0. Finally, Figure~\ref{fig: Result asymptotic t}-right shows that pair $(\bar \beta, \bar R)$, where $\bar \beta = -\bar q_1/2\bar q_2$ satisfies constraint~\eqref{equa: con2}. As for constraint~\eqref{equa: con1}, we have that $x(p,R_t)-y(p,\beta_t,R_t) \rightarrow 1$ as $t\rightarrow \infty$, which was implied by that the discriminant of $\bar q_0 + \bar q_1 \beta + \bar q_2 \beta^2$ is $0$.

\def\length{0.3}
 \begin{figure}[h!]
\centering
  \includegraphics[width=\length\linewidth]{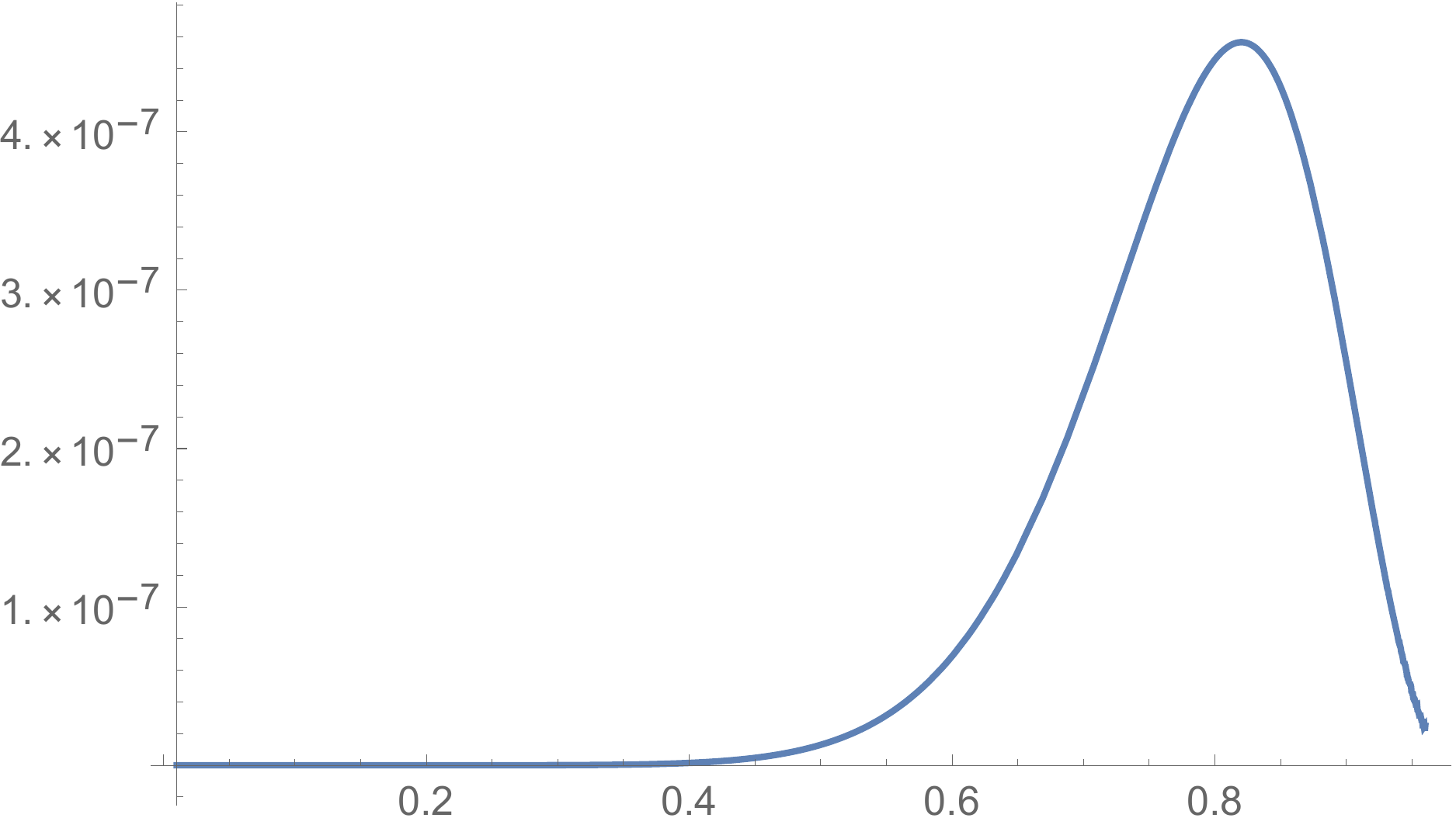}
~~~ \includegraphics[width=\length\linewidth]{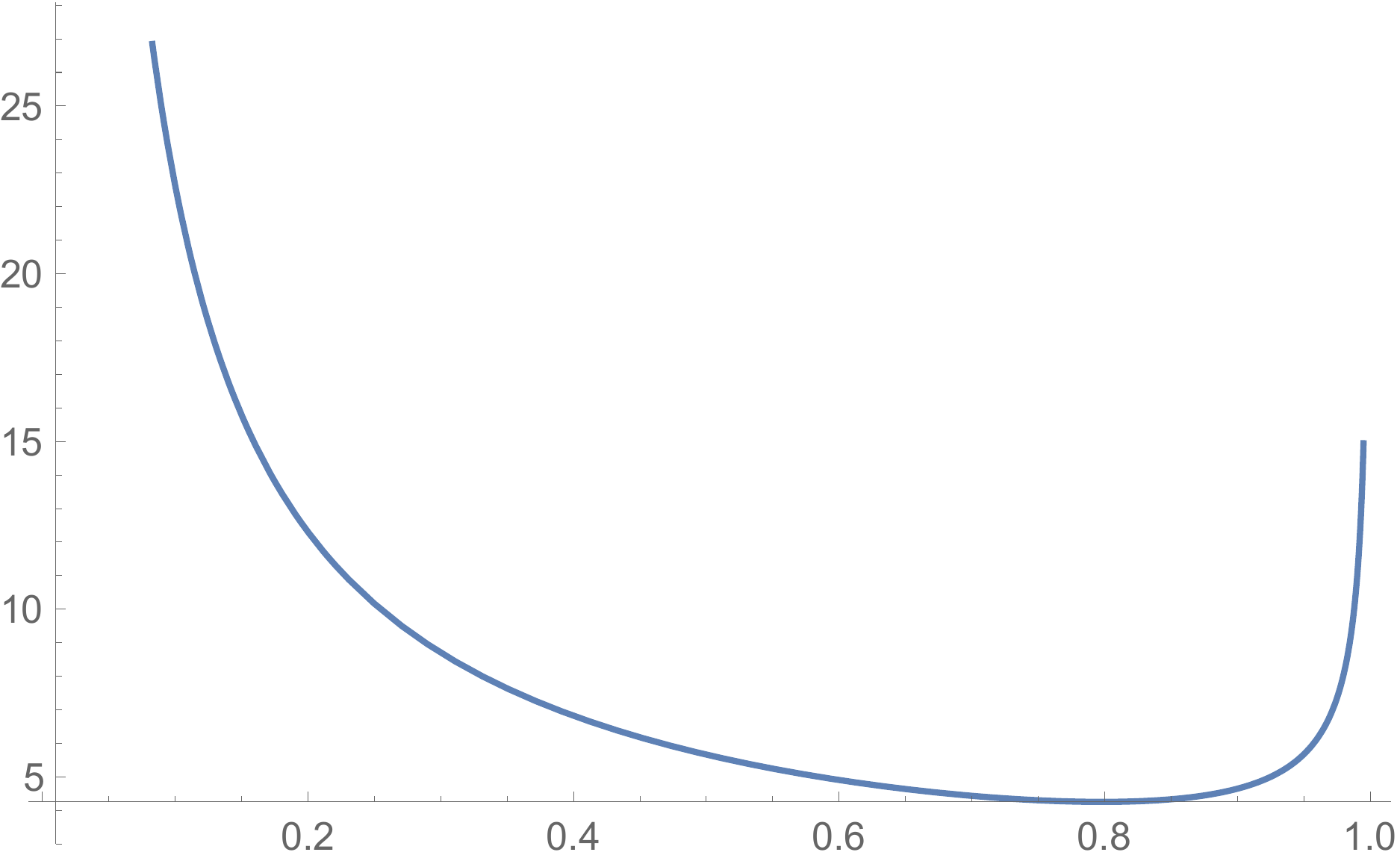}
~~~~  \includegraphics[width=\length\linewidth]{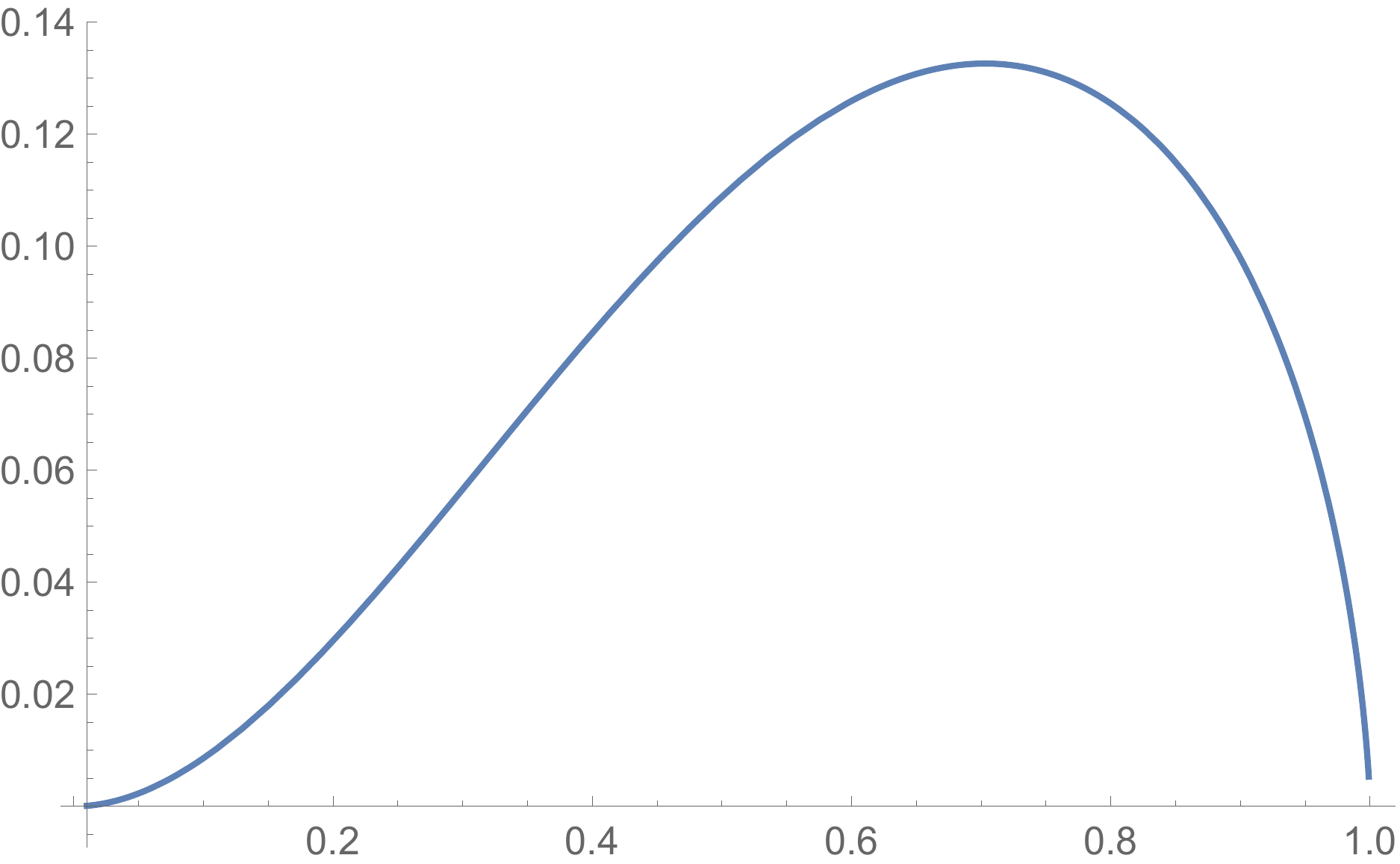}
\caption{
Figures summarize properties of the limit of pair $(\beta_t,R_t)$, as defined by Theorem~\ref{thm: best t-sub-origin}, and as $t$ tends to infinity. All the horizontal axes correspond to $p\in (0,1)$. 
- Left figure displays the behavior of $R_{10}- \bar R$, that is the difference between the achieved competitive ratio of the $10$-sub-monotone algorithm of Theorem~\ref{thm: best t-sub-origin} and the ratio $\bar R$ one can achieve for arbitrary large values of $t$. 
- Middle figure displays the behavior $x(p,\bar R)$, as a function of $p$, according to which $x(p,\bar R)>4$ for all $p\in (0,1)$, and hence, it is bounded away from 1 as wanted. 
- Right figure shows that the pair $(\bar \beta, \bar R)$ satisfies constraint~\eqref{equa: con2}. The vertical axis corresponds to $\left( \bar \beta - \frac{E}{\bar R/p-F} \right)(1-p)^2$.
}
\label{fig: Result asymptotic t}
\end{figure}

\section{Discussion and Open Problems}
\label{sec: discussion}
We studied \textit{$p$-Faulty Search} (\FS), a search problem on a $1$-ray, where the searcher is probabilistically faulty with known probability $1-p$. Our main contribution pertains to the disproof of a conjecture that optimal trajectories for such problems are monotone. Whether the same conjecture is wrong for searching $m$-rays, and in particular, the line ($m=2$) remains an open problem. 
When it comes to searching the half-line, all our algorithms have competitive ratio at least 4 when $p\rightarrow 0$ and at least 3 when $p\rightarrow 1$. The value of 3 is provably a lower bound to any search strategy since the searcher has to return at least once close to the 
origin before attempting for a second time an expansion of the searched space. No other general lower bound is known for the problem, whereas all our algorithms have competitive ratio at least $4-p$. Is $4-p$ a lower bound to any algorithm for \FS, and if yes can this be matched by an upper bound? We conjecture that the lower bound is valid, as well as that our $t$-sub-monotone algorithms are sub-optimal.

\section*{Acknowledgements}
The authors would like to thank Huda Chuangpishit, Sophia Park, Bhargav Parsi and Benjamin Reiniger for many fruitful discussions.

%
%
%
%
%
%
 \bibliographystyle{plain}
  \bibliography{refs}

\begin{showappendix}

\section*{Appendix}

\appendix

\section{Lemma~\ref{lem: exp term incr}}

\begin{proof}
\ignore{
Note that for each $i$ we have $f_i=\sum_{j=1}^ig_i$. Hence,
\begin{align*}
\cost{T}(d)
=& \sum_{i=1}^\infty p(1-p)^{i-1}f_i \\
=& p \sum_{i=1}^\infty (1-p)^{i-1}\sum_{j=1}^ig_i \\
=& p \sum_{i=1}^\infty g_i \sum_{j=i-1}^\infty(1-p)^{j-1} \\
=& p \sum_{i=1}^\infty (1-p)^{i-1}g_i \sum_{j=0}^\infty(1-p)^{j} \\
=&\sum_{i=1}^\infty (1-p)^{i-1}g_i.
\end{align*}
}
Note that for each $i$ we have $f_i=\sum_{j=1}^ig_j$. We then have that
\begin{align*}
\cost{T}(d)
=& \sum_{i=1}^\infty p(1-p)^{i-1}f_i \\
=& p \sum_{i=1}^\infty (1-p)^{i-1}\sum_{j=1}^ig_j \\
=& p \sum_{j=1}^\infty g_j \sum_{i=j}^\infty(1-p)^{i-1} \\
=& p \sum_{j=1}^\infty (1-p)^{j-1}g_j \sum_{i=0}^\infty(1-p)^{i} \\
=&\sum_{j=1}^\infty (1-p)^{j-1}g_j,
\end{align*}
and the proof follows.
 \end{proof}

\section{Lemma~\ref{lem: worst placement of treasure, backtoorigin}}
\begin{proof}
Suppose that the treasure is located at point $d=x_r + y \in (x_r, x_{r+1})$, where $0<y<x_{r+1}-x_r$. With that notation in mind (see also Figure~\ref{fig: BackToOrigin}), we compute the time intervals $g_i$ between consecutive visitations, as they were defined in Lemma~\ref{lem: exp term incr}. We have that
\begin{align*}
g_1 &= 2\sum_{i=1}^r x_i + x_r +y = 2\sum_{i=1}^r x_i + d\\
g_{2i} &= 2(x_{r+i}-x_r-y)= 2(x_{r+i}-d), \qquad ~~i=1,\ldots, \infty \\
g_{2i+1} &= 2x_r+2y=2d, \qquad ~~i=1,\ldots, \infty.
\end{align*}

\begin{figure}[!ht]
\vspace{-0.1in}
                \centering
                \includegraphics[scale=0.7]{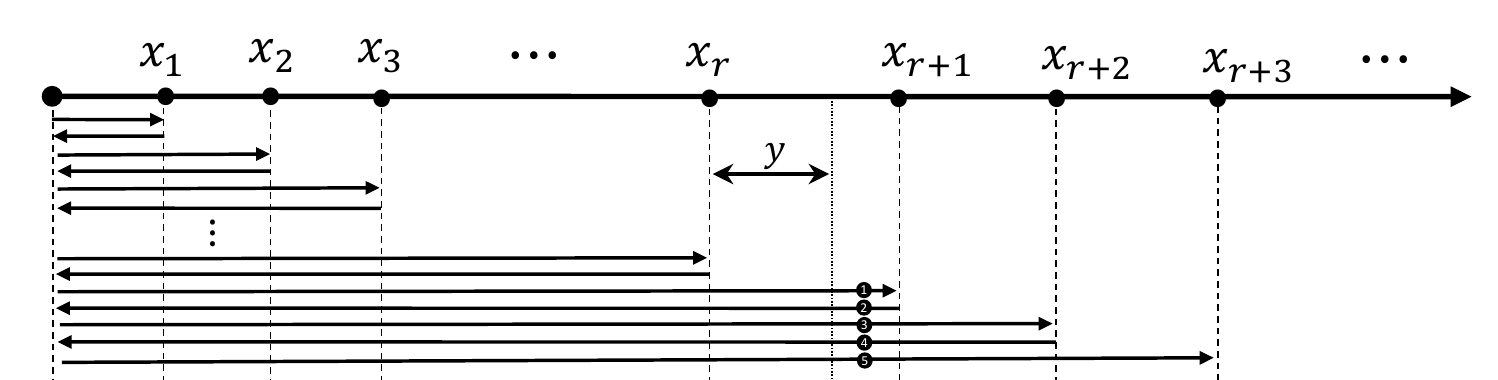}
                \vspace{-0.1in}
\caption{Monotone algorithm $\{x_i\}_{i\geq 1}$. Figure also depicts the first 5 visitations of the treasure that is placed at  $x_r+y$.
}
                \label{fig: BackToOrigin}
\vspace{-0.1in}
\end{figure}

Therefore, by Lemma~\ref{lem: exp term incr} the expected termination time $\cost{T}(d)$ for algorithm $T$ is 
\begin{align*}
\sum_{i=1}^\infty (1-p)^{i-1}g_i
&= g_1
+ \sum_{i\geq 1}(1-p)^{2i-1}g_{2i}
+ \sum_{i\geq 1}(1-p)^{2i}g_{2i+1} \\
&=
\left( 2\sum_{i=1}^r x_i + d \right)
+
2\left( \sum_{i\geq 1}(1-p)^{2i-1}(x_{r+i}-d)
 \right)
+
2d\left( \sum_{i\geq 1}(1-p)^{2i}
 \right) \\
&=
2 \sum_{i=1}^r x_i
+ 2 \sum_{i\geq 1}(1-p)^{2i-1}x_{r+i}
+ d \left( 1 -2p \sum_{i\geq 1}(1-p)^{2i-1} \right)\\
&=2 \sum_{i=1}^r x_i
+ 2 \sum_{i\geq 1}(1-p)^{2i-1}x_{r+i}
+ d \frac{p}{2-p}.
\end{align*}
Recall that the competitive ratio of this algorithm is $p \cost{T}(d)/d$, and hence, in the worst case, $d$ approaches $x_r$ from the right.
 \end{proof}

\section{Lemma~\ref{lem: upper bound backtoorigin}}
\begin{proof}
We study the restricted family of monotone trajectories $T=\{x_i\}_{\geq 1}$, where $x_i=b^i$, for some $b=b(p)$.
By Lemma~\ref{lem: worst placement of treasure, backtoorigin}, the competitive ratio of search strategy $T$ is at most
\begin{align}
\sup_r \left\{
2\frac p{b^r} \sum_{i=1}^r b^i
+ 2 \frac p{b^r} \sum_{i\geq 1}(1-p)^{2i-1}b^{r+i}
+  \frac{p^2}{2-p}
\right\}
=&
\sup_r \left\{
p\frac{2b \left(b^r-1\right)}{b^{r}(b-1)}
+
p\frac{2b (1-p)}{1-b (1-p)^2}
+
\frac{p^2}{2-p}\right\} \notag \\
=&\lim_{r\rightarrow \infty} \left\{
p\frac{2b \left(b^r-1\right)}{b^{r}(b-1)}
+
p\frac{2b (1-p)}{1-b (1-p)^2}
+
\frac{p^2}{2-p}\right\} \notag \\
=&
p\frac{2b}{b-1}
+
p\frac{2b (1-p)}{1-b (1-p)^2}
+
\frac{p^2}{2-p}. \label{equa: comp ratio wrt b origin}
\end{align}
Calculations above assume that 
$b<1/(1-p)^2$, as otherwise, the second summation is divergent.
We will make sure later that our choice of $b$ complies with this condition.
Note also that for $x_i$ to be increasing, we need $b>1$.
Now, denote expression~\eqref{equa: comp ratio wrt b origin} by $f(b)$. We will determine the choice of $b$ that minimizes $f(b)$, given that $1<b<1/(1-p)^2$.

It is straightforward to see that $\frac{d^2}{db^2}f(b)=4p \left(\frac{(1-p)^3}{1-\left(b (p-1)^2\right)^3}+\frac{1}{(b-1)^3}\right)$, and hence, $f(b)$ is convex when $b\in \left(1, 1/(1-p)^2\right)$. Hence, if $\frac{d}{db}f(b)$ has a root in $\left(1, 1/(1-p)^2\right)$, that would be a minimizer. Indeed, $$\frac{d}{db}f(b)=2p\left(\frac{1-p}{\left(1-b (1-p)^2\right)^2}-\frac{1}{(b-1)^2}\right)$$ has two roots $\frac{1}{\sqrt{1-p}\left(\pm(2-p)-\sqrt{1-p}\right)}$, one being positive and one negative (for all values of $p\in (0,1)$). We choose the positive root, that we call $b_p$, 
and it is elementary to see that $1<b_p<1/(1-p)^2$, for all $p\in (0,1)$, as wanted. Substituting $b=b_p$ in~\eqref{equa: comp ratio wrt b origin} gives the competitive ratio promised by the statement of the lemma. 
\end{proof}

\section{Lemma~\ref{lem: condition on alphas}}

\begin{proof}
If the treasure is placed arbitrarily close to turning point $f_k$, then by Lemma \ref{lem: exp term incr}, a lower bound to the best possible competitive ratio $c$ satisfies the following (infinitely many) constraints: 
$$c\geq \frac{p}{f_k}\sum_{i=1}^{\infty}(1-p)^{i-1}g_i^k, \quad k=0,\ldots, \infty.$$
We next restrict our attention to the first $\ell+1$ such constraints, where $\ell$ is an arbitrary integer. Hence, we require that
$$c\geq \frac{p}{f_k}\sum_{i=1}^{\infty}(1-p)^{i-1}g_i^k, \quad k=0,\ldots,\ell.$$
Now, multiply both hand-sides of the inequalities by $f_k/p$ to obtain 
\begin{align*}
f_k\frac{c}{p}\geq \sum_{i=1}^{\infty}(1-p)^{i-1}g_i^k &=2\sum_{i=1}^k{f_i}+f_k+2\sum_{i=1}^{\infty} (1-p)^{2i-1}(f_{k+i}-f_k)+2\sum_{i=1}^{\infty}(1-p)^{2i}f_k\\
&\geq 2\sum_{i=1}^k{f_i}+f_k+2\sum_{i=1}^{\ell-k}(1-p)^{2i-1}(f_{k+i}-f_k) \\
&~~~+2\sum_{i=\ell-k+1}^{\infty}(1-p)^{2i-1}(f_\ell-f_k)+2f_k\sum_{i=1}^{\infty}(1-p)^{2i}\\
&=2\sum_{i=1}^{k-1}f_i+f_k\left(3-2\sum_{i=1}^{\infty}(1-p)^{2i-1}+2 \sum_{i=1}^{\infty}(1-p)^{2i}\right)\\
&~~~+2\sum_{i=k+1}^{\ell}(1-p)^{2(i-k)-1}f_i+2f_\ell\frac{(1-p)^{2(\ell-k)+1}}{p(2-p)}.
\end{align*}
We conclude that $f_k c / p$ is at least the last term above, so after rearranging the terms of the inequality, bringing them all on one side, and factoring out the $f_i$ terms, we have that
\begin{align*}
\sum_{i=0}^{k-1}f_i+\left(\frac{1}{2}+\frac{1}{2-p}-\frac{c}{2p}\right) f_k+\sum_{i=k+1}^{\ell}(1-p)^{2(i-k)-1}f_i+\frac{(1-p)^{2(\ell-k)+1}}{p(2-p)}f_\ell \leq 0,
\end{align*}
as desired.
\end{proof}

\section{Lemma~\ref{lem: monotonicity implies bound}}
\begin{proof}
We proceed by finding a closed formula for $f_{\ell-1}$ and then imposing monotonicity. Our first observation is that for all $0\leq k\leq \ell-1$ we have that $\gamma_{\ell,k}+\beta_{\ell,k}=\frac{(1-p)^{2(\ell-k)-1}}{(2-p)p}$. Setting $r:=\frac{1}{(2-p)p}$ allows us to  rewrite the matrix of system~\eqref{equa: optimal turning points} as
\[
  A_\ell =
  \left[ {\begin{array}{ccccc}
   (1-p)&(1-p)^3&(1-p)^5&\ldots &r(1-p)^{2\ell-1}\\
   \alpha & (1-p) &(1-p)^3&\ldots & r(1-p)^{2\ell-3} \\
   1& \alpha & (1-p)&\ldots & r(1-p)^{2\ell-5} \\
   1&1&\alpha & \ldots & r(1-p)^{2\ell-7} \\
     \vdots&\vdots&\vdots & \ldots &\vdots \\
      1&1&1 & \ldots & r(1-p) \\
   \end{array} } \right].
\]
We proceed by applying elementary row operations to the system. From each row of $A_\ell$ (except the last one) we subtract a $(1-p)^2$ multiple of the following row to obtain linear system $\bar{A_\ell} f = \bar{b}$, where

\[
\bar{A_\ell}=
  \left[ {\begin{array}{cccccc}
   1-p-\alpha  (1-p)^2&0 &0&\ldots &0&0\\
   \alpha-(1-p)^2 & 1-p-\alpha  (1-p)^2 &0&\ldots & 0 & 0\\
   1-(1-p)^2& \alpha-(1-p)^2 & 1-p-\alpha  (1-p)^2&\ldots & 0 & 0 \\
   1-(1-p)^2&1-(1-p)^2&\alpha-(1-p)^2 & \ldots & 0 &0 \\
     \vdots&\vdots&\vdots & \ldots &\vdots&\vdots \\
      1-(1-p)^2&1-(1-p)^2&1-(1-p)^2 & \ldots &1-p-\alpha  (1-p)^2&0 \\
        1&1&1& \ldots &\alpha&r(1-p)\\
   \end{array} } \right]
\]
and $\bar{b}^T =
(-\alpha+(1-p)^2,
 -1+(1-p)^2,
 -1+(1-p)^2,
   -1+(1-p)^2,
   \hdots,
   -1+(1-p)^2,
   -1)
   $.
Now set
$$s:=1-p-\alpha  (1-p)^2, ~~t:=\alpha-(1-p)^2, ~~w:=1-(1-p)^2,$$
and define $\ell \times \ell$ matrix
\[
\ignore{
  D_\ell :=
  \left[ {\begin{array}{ccccccc}
   s&0 &0&\ldots &0&0&-t\\
   t & s &0&\ldots &0& 0 & -w\\
 w& t & s&\ldots & 0 &0& -w \\
   w&w&t & \ldots & 0 &0&-w \\
     \vdots&\vdots&\vdots & \ldots &\vdots&\vdots&\vdots \\
      w&w&w& \ldots &s&0&-w \\
        w&w&w& \ldots &t&s&-w \\
        1&1&1& \ldots &1&\alpha&-1 \\
   \end{array} } \right],
}
 C_{\ell}:=
  \left[ {\begin{array}{ccccccc}
   s&0 &0&\ldots &0&-t&0\\
   t & s &0&\ldots &0& -w&0\\
 w& t & s&\ldots & 0 &-w&0 \\
   w&w&t & \ldots & 0 &-w&0 \\
     \vdots&\vdots&\vdots & \ldots &\vdots&\vdots&\vdots \\
      w&w&w& \ldots &s&-w&0 \\
        w&w&w& \ldots &t&-w&0 \\
        1&1&1& \ldots &1&-1&r(1-p) \\
   \end{array} } \right].
\]
By Cramer's rule we have that
$$
f_{\ell-1}=\frac{\mathrm{det}(C_{\ell})}{\mathrm{det}(A_\ell)}.
$$
Note that $\mathrm{det}(A_\ell)= \left( 1-p-\alpha  (1-p)^2 \right)^{\ell-1}r(1-p)$.

Next we compute $\mathrm{det}(C_{\ell})$. We denote the $(\ell-1)\times(\ell-1)$ principal minor of $C_\ell$ as $B_{\ell-1}$. The last row of $B_{\ell-1}$ is $(w,w,\ldots,w,t,-w)$. We further denote by $L_{\ell-1}$ the matrix we obtain from $B_{\ell-1}$ by scaling its last row by $w$ so that it reads $(1,1,\ldots,1,t/w,-1)$. Finally, we denote by $K_{\ell-1}$ the matrix we obtain by replacing the last row of $B_{\ell-1}$ by $(1,1,\ldots, 1,1,-1)$; that is, the all-1 row except from the last entry which is -1. With this notation in mind, we note that
$$
\det(C_\ell) = -r(1-p)\det(B_{\ell-1}) =- \frac{r(1-p)}{w}\det(L_{\ell-1}).
$$
Now expanding the determinants of $K_{\ell-1}, L_{\ell-1}$ with respect to their first rows we obtain the system of recurrence equations
\begin{align*}
\det(K_{\ell-1}) &= s \det(K_{\ell-2}) - w \det(L_{\ell-2}), \\
\det(L_{\ell-1}) &= s \det(K_{\ell-2}) - t \det(L_{\ell-2}).
\end{align*}
We solve the first one with respect to $\det(L_{\ell-2})$ and we substitute to the second one to obtain the following recurrence exclusively on $K_\ell$
$$
\det(K_\ell) + (t-s)\det(K_{\ell-1}) + s(w-t) \det(K_{\ell-2}) = 0.
$$
The characteristic polynomial of the latter degree-2 linear recurrence has discriminant equal to
$$
(t-s)^2 - 4s(w-t)
=
\frac{1}{4} \left((2-p)^2  c^2 +2 ((p-2) p+4) (p-2) c +p^2 ((p-4) p+12)\right),
$$
which in particular is a degree-2 polynomial $g(c)$ in the competitive ratio $c$ and has discriminant $4 (2-p)^2 (1-p)$. Since $g(c)$ is convex, we conclude that the discriminant of the characteristic polynomial is non-negative when $c$ is larger than the largest root of $g(c)$, that is when
$$
c\geq
\frac{
 (4-(2-p) p) (2-p) + 4(2-p) \sqrt{1-p}
}{(2-p)^2}
=
\frac{4+4\sqrt{1-p}}{2-p}-p,
$$
and the proof follows. 
\end{proof}

\section{Theorem~\ref{thm: summary of worst case in intervals}}

We need a few lemmata before we proceed with the proof of Theorem~\ref{thm: summary of worst case in intervals}.

\begin{lemma}\label{lem:hop}
For any $j$, the time $h_j$ required for the $t$-hop $x_j \rightarrow x_{j+1}$ is
$$
h_j:=\beta^j \left( \beta+2\gamma_t-3 \right)
= \left( \beta^{j+1}-\beta^j\right) \frac{\beta+2\gamma_t-3}{\beta-1}.
$$
\end{lemma}

\begin{proof}
The reader may consult Figure~\ref{fig: General-SubOrigin}. The interval is traversed exactly three times, except from the interval $[\gamma_t\beta^r,\beta^r]$ which is traversed once. Hence, the time for a robot to move from $\beta^j$ to $\beta^{j+1}$ is
$$
3\left( \beta^{j+1}-\beta^j\right) -
2\left( \beta^{j+1}-\gamma_t \beta^j\right)
=\beta^j \left( \beta+2\gamma_t-3 \right).
$$
The alternative expression is obtained by factoring out $\left( \beta^{j+1}-\beta^j\right)$ and is given for convenience.
\end{proof}

Using the above, we compute the total time the robot needs to progress from the origin to $\beta^r+\epsilon$.

\begin{lemma}\label{Lem:TimeToReachbr}
For any sufficiently small $\epsilon>0$, the time needed for the robot to reach $\beta^r +\epsilon$ for the first time is equal to
$$
\beta^{r}\frac{3\beta+2\gamma_t-3}{\beta-1}
-\frac{2 \beta \gamma_t}{\beta-1}
+\epsilon.
$$
\end{lemma}

\begin{proof}
The algorithm will perform a number of hops before returning to the origin after each hop. According to Lemma~\ref{lem:hop}, the total time for this trajectory is
\begin{align*}
3\beta + \sum_{j=1}^{r-1}h_j + 2\sum_{j=1}^{r-1} \beta^{j+1}
&=
3\beta + \left( \beta^{r}-\beta\right) \frac{\beta+2\gamma_t-3}{\beta-1}
+ 2\sum_{j=1}^{r-1} \beta^{j+1} \\
&=
3\beta + \left( \beta^{r}-\beta\right) \frac{\beta+2\gamma_t-3}{\beta-1}
+ 2\frac{\beta \left(\beta^r- \beta \right)}{\beta-1} \\
&=
\beta^{r}\frac{3\beta+2\gamma_t-3}{\beta-1}
-\frac{2 \beta \gamma_t}{\beta-1},
\end{align*}
and the proof follows.
\end{proof}

The proof of Theorem~\ref{thm: summary of worst case in intervals} is given by Lemmas~\ref{lem: CR cases Ai}, \ref{lem: CR cases At+1} at the end of the current section. Towards establishing the lemmas, we need to calculate the time between consecutive visitations of the treasure in order to eventually apply Lemma~\ref{lem: exp term incr} and compute the performance of a $t$-sub-monotone algorithm.

As we did previously and for the sake of simplifying the analysis, we assume that the treasure will never coincide with a turning point $\gamma_ix_j$.
Moreover, we assume that the treasure is placed at distance $d_i=\gamma_{i-1} x_r + y$ from the origin, where $0<y<(\gamma_{i}-\gamma_{i-1}) x_r$, for some $i$ that we allow for the moment to vary.

Since the treasure can be in any of these intervals, there are $t+1$ cases to consider when computing the performance of the algorithm.
Lemmas~\ref{lem: subvisitations general suborigin} and \ref{lem: subvisitations general suborigin last case} concern different cases as to where the treasure is with respect to internal turning points associated with $\gamma_i$.

\begin{lemma}\label{lem: subvisitations general suborigin}
For any $i=1,\ldots, t$, suppose that the treasure is placed at distance $d_i=\gamma_{i-1} \beta^r + y$ from the origin, where $0<y<(\gamma_{i}-\gamma_{i-1}) x_r$. We then have that
$$
g_s =
\left\{
\begin{array}{ll}
\beta^{r}\left(  \frac{2\gamma_t}{\beta-1} -\gamma_i+ 3\gamma_{i-1}  \right) -\frac{2 \beta \gamma_t}{\beta-1}  +d_i, 	&\textrm{if}~s=1 \\
2\gamma_i \beta^r - 2d_i, 	&\textrm{if}~s=2 \\
2y, 	&\textrm{if}~s=3 \\
2\beta^r\left( \beta+\gamma_t\right) -4d_i, 	&\textrm{if}~s=4 \\
2d_i, 	&\textrm{if}~s=2j+3~\textrm{for some}~j\geq 1 \\
2\beta^{r+j}\left( \beta+\gamma_t-1 \right) -2d_i, 	&\textrm{if}~s=2j+4~\textrm{for some}~j\geq 1 \\
\end{array}
\right.
$$
\end{lemma}

\begin{proof}
For computing each of the $g_j$'s we consult Figure~\ref{fig: General-SubOrigin}.
\begin{align}
g_1
&=
\beta^{r}\frac{3\beta+2\gamma_t-3}{\beta-1} -\frac{2 \beta \gamma_t}{\beta-1}
+ 3(\gamma_{i-1}-1)\beta^r +y  			\tag{By Lemma~\ref{Lem:TimeToReachbr}} \\
&=\beta^{r}\left(  \frac{3\beta+2\gamma_t-3}{\beta-1} + 3(\gamma_{i-1}-1) \right)
-\frac{2 \beta \gamma_t}{\beta-1}  +y  \notag \\
&=\beta^{r}\left(  \frac{2\gamma_t}{\beta-1} + 3\gamma_{i-1}  \right)
-\frac{2 \beta \gamma_t}{\beta-1}  +y.  \notag
\end{align}

We derive that $g_2 = 2\gamma_i \beta^r - 2d_i$, that $g_3=2y$, and that
\begin{align}
g_4
&= 4\left( \gamma_t \beta^r-d_i\right)+2\left(\beta^{r+1}-\gamma_t\beta^r\right) \notag \\
&= 2\gamma_t\beta^r + 2\beta^{r+1}-4d_i  \notag \\
&=2\beta^r\left( \beta+\gamma_t\right) -4d_i. \notag
\end{align}
After the fourth visitation of the treasure, an odd indexed visitation takes time $2d_i$; that is, $g_{2j+3}=2d_i$, for all $j\geq 1$.
Finally, for every even indexed visitation after the 4th one we have, for each $j\geq 1$, that
\begin{align}
g_{2j+4}
&=
\left( \beta^{r+j}
 -d_i \right) + h_{r+j}+\left( \beta^{r+j+1}-d_i \right)
\notag \\
&=
 \beta^{r+j}+\beta^{r+j+1}
+ \beta^{r+j} \left( \beta+2\gamma_t-3 \right)
 -2d_i
 \tag{by Lemma~\ref{lem:hop}} \\
 &= 2\beta^{r+j}\left( \beta+\gamma_t-1 \right) -2d_i, \notag
\end{align}
and the proof follows. 
\end{proof}

\begin{lemma}\label{lem: subvisitations general suborigin last case}
Suppose that the treasure is placed at distance $d_{t+1}=\gamma_{t} \beta^r + y$ from the origin, where $0<y<(\beta-\gamma_{t}) x_r$. We then have that
$$
g_s =
\left\{
\begin{array}{ll}
\beta^{r}\left(  \frac{2\gamma_t}{\beta-1} + 3\gamma_{t}  \right) -\frac{2 \beta \gamma_t}{\beta-1}  +y, 	&\textrm{if}~s=1 \\
2 \beta^{r+1} - 2d_{t+1}, 	&\textrm{if}~s=2 \\
2d_{t+1}, 	&\textrm{if}~s=2j+1~\textrm{for some}~j\geq 1 \\
2\beta^{r+j}\left( \beta+\gamma_t-1 \right) -2d_{t+1}, 	&\textrm{if}~s=2j+2~\textrm{for some}~j\geq 1
\end{array}
\right.
$$
\end{lemma}

\begin{proof}
For the first two visitations, the time elapsed is identical to the case where the treasure is in any of the intervals $A_i$ (see Figure~\ref{fig: General-SubOrigin}). We only need to set $i=t+1$, in which case, by Lemma~\ref{lem: subvisitations general suborigin} we obtain $g_1,g_2$ as claimed (recall that $\gamma_{t+1}=\beta$). Any odd visitation thereafter will take additional time $2d_{t+1}$. Finally, every even visitation thereafter is identical to the (large indexed) even visitations of Lemma~\ref{lem: subvisitations general suborigin}, only that in the currently examined case, the index of the visitations starts from four, instead of six.
\end{proof}

We are now ready to prove Theorem~\ref{thm: summary of worst case in intervals} by proposing and proving Lemmas~\ref{lem: CR cases Ai}, \ref{lem: CR cases At+1}, each of them describing the worst case competitive ratio over all possible placements of the treasure.

\begin{lemma}\label{lem: CR cases Ai}
For any $i=1,\ldots,t$, and given that the treasure lies in interval $A_i$, the worst case induced competitive ratio $R_i$ is given by the formula
$
R_i =p \left( \tfrac{ A\gamma_i + B \gamma_t + C}{\gamma_{i-1}} +D \right)$.
\end{lemma}

\begin{proof}
Suppose that the treasure is placed at distance $d_i=\gamma_{i-1} \beta^r + y$ from the origin, where $0<y<(\gamma_{i}-\gamma_{i-1}) x_r$. Let $C_i$ denote the expected termination time in this case. As per Lemma~\ref{lem: exp term incr}, we have that
$
C_i = \sum_{j=1}^\infty (1-p)^{i-1}g_j,
$
and recall that the competitive ratio in this case will be given by
$
p~\sup_{y,r} \frac{C_i}{d_i} = p~ \sup_{y,r} \frac{C_i}{\gamma_{i-1} x^r + y}.
$
From the above and Lemma~\ref{lem: subvisitations general suborigin} it is immediate that the largest competitive ratio is induced when $y\rightarrow 0$ (and as it will be clear momentarily, when $r\rightarrow \infty$). Therefore, in what follows we use $d_i=\gamma_{i-1}\beta^r$. We then have that

\begin{align*}
\frac{C_i}{d_i}
=&
\frac1{d_i}\left(
g_1 + (1-p)g_2 +(1-p)^3g_4 + \sum_{j\geq 1}(1-p)^{2j+2}g_{2j+3}+\sum_{j\geq 1}(1-p)^{2j+3}g_{2j+4}
\right)
\\
=&
\frac1{\gamma_{i-1}}\left(  \frac{2\gamma_t}{\beta-1} + 3\gamma_{i-1}  \right) -\frac{2 \beta \gamma_t}{\gamma_{i-1}\beta^{r}(\beta-1)}
+
2(1-p)\left( \frac{\gamma_i}{\gamma_{i-1}} - 1 \right) \\
& + 2(1-p)^3
\left(
\frac{\beta+\gamma_t}{\gamma_{i-1}} -2
\right)
+2 \sum_{j\geq 1}(1-p)^{2j+2} \\
& + \frac2{\gamma_{i-1}}\left( \beta+\gamma_t-1 \right) \sum_{j\geq 1}(1-p)^{2j+3}\beta^j
-2 \sum_{j\geq 1}(1-p)^{2j+3} \\
\stackrel{(r\rightarrow \infty)}{\leq}&
\ignore{
\frac2{\gamma_{i-1}}
\left(
 \frac{\gamma_t}{\beta-1}
 +(1-p)\gamma_i
 +(1-p)^3(\beta+\gamma_t-\gamma_i)
 +\left( \beta+\gamma_t-1 \right)\frac{\beta (1-p)^5}{1-\beta(1-p)^2}
\right) \\
&+ 3-2(1-p)-2(1-p)^3
+\frac{2 (1-p)^4}{2-p} \\
=&
}
\frac2{\gamma_{i-1}}
\left(
(1-p)\gamma_i +
+\left(\frac1{\beta-1}+\frac{(1-p)^3}{1-\beta(1-p)^2} \right)\gamma_t
+ \frac{p(1-p)^3(2-p)\beta}{1-\beta(1-p)^2}
\right) \\
&+
\frac{-2 p^4+12 p^3-26 p^2+23 p-4}{2-p},
\end{align*}
and the proof follows.
\end{proof}

\begin{lemma}\label{lem: CR cases At+1}
Given that the treasure lies in interval $A_{t+1}$, the worst case induced competitive ratio $R_{t+1}$ is given by the formula
$R_{t+1} = p \left(\tfrac{E}{\gamma_t}+F \right)$.
\end{lemma}

\begin{proof}
We invoke  Lemma~\ref{lem: exp term incr}, which together with Lemma~\ref{lem: subvisitations general suborigin last case} allows us to compute the expected termination time $C_{t+1}$. Calculations are similar to the proof of Lemma~\ref{lem: CR cases Ai}, and in particular, the worst competitive ratio $R_{t+1}$ is induced when $y\rightarrow 0$; that is, when $d_{t+1}\rightarrow \gamma_t \beta^r$, and when $r\rightarrow \infty$. More specifically,
\begin{align*}
\sup_{r,y}
\frac{C_{t+1}}{d_{t+1}}
=& \sup_{r,y}
\frac1{d_{t+1}}\left(
g_1 + (1-p)g_2  +\sum_{j\geq 1}(1-p)^{2j}g_{2j+1}+\sum_{j\geq 1}(1-p)^{2j+1}g_{2j+2}
\right)
\\
=&
\left(  \frac{2}{\beta-1} + 3 \right)
+2 (1-p) \left( \frac{\beta}{\gamma_t}-1 \right) \\
&
+ 2 \sum_{j\geq 1}(1-p)^{2j}
+ 2\frac{ \beta+\gamma_t-1}{\gamma_{t}} \sum_{j\geq 1}(1-p)^{2j+1}\beta^j
-2 \sum_{j\geq 1}(1-p)^{2j+1} \\
=&
\frac2{\gamma_t}
\frac{p(1-p)(2-p)\beta}{1-\beta(1-p)^2}
+p \left(2 \frac{ \beta(1-p)+1}{(\beta-1)\left(1-\beta(1-p)^2\right)}+ \frac{5-2 p}{2-p} \right),
\end{align*}
and the proof follows.
\end{proof}

\section{Theorem~\ref{thm: best t-sub-origin}}
The main ingredient for proving Theorem~\ref{thm: best t-sub-origin} is the following lemma.

\begin{lemma}
\label{lem: best t-sub-origin}
For some $p \in (0,1)$, consider values of $t,R,\beta$ 
satisfying constraint~\eqref{equa: nonlinearequality}. 
\ignore{
\begin{equation}
\label{equa: nonlinearequality}
\left( 1-\frac{y}{x-1}\right) x^t + \frac{y}{x-1} - \frac{E}{R/p-F}=0.
\end{equation}
}
If additionally, the pair $(\beta,R)$ is feasible, then $R$ is the competitive ratio of a $t$-sub-monotone trajectory with parameters $\beta, \gamma_1, \ldots, \gamma_t$ for problem \FS, where
$
\gamma_i = \left( 1-\frac{y}{x-1}\right) x^i + \frac{y}{x-1},~i=1\ldots, t
$.
\end{lemma}

\ignore{
The solution to the optimization problem is then
\begin{align*}
\inf_{R,\beta, \gamma_i} ~~~& R \\
\textrm{s.t.}~~ &
\left( 1-\frac{y}{x-1}\right) x^t + \frac{y}{x-1} = \frac{E}{R/p-F} \\
& x-y >1 \\
& \frac{E}{R/p-F} < \beta,
\end{align*}
gives the competitive ration $R$ of a $t$-sub-monotone trajectory with parameters $\beta, \gamma_1, \ldots, \gamma_t$.
}

\begin{proof}
By Theorem~\ref{thm: summary of worst case in intervals}, the best $t$-sub-monotone algorithm is determined by parameters $\gamma_1, \gamma_2, \ldots, \gamma_t, \beta$ that minimize $\max \left\{ R_1, R_2, \ldots, R_t, R_{t+1} \right\}$, subject to that $1<\gamma_1 <\ldots < \gamma_t <\beta < \frac1{(1-p)^2}$.
The bound on $\beta$ guarantees convergence of the expected termination time. We attempt to find a solution to the optimization problem above by requiring that
$$
R_1 = R_2 = \ldots = R_t = R_{t+1}.
$$
Denote the value of the optimal solution by $R$, and suppose that it is realized by parameters $\gamma_1, \gamma_2, \ldots, \gamma_t, \beta$.
By Lemma~\ref{lem: CR cases At+1}, we have that
\begin{equation}\label{equa: opt gt}
\gamma_t = \frac{E}{R/p-F}
\end{equation}
We then have that by Lemma~\ref{lem: CR cases Ai} and solving for $\gamma_i$ we obtain that for each $i=1,\ldots, t$
\begin{align*}
\gamma_i &=\frac{R/p-D}{A}\gamma_{i-1} - \frac{B\gamma_t +C}{A}\\
			&\stackrel{\eqref{equa: opt gt}}{=} \frac{R/p-D}{A}\gamma_{i-1} - \frac{\frac{B~E}{R/p-F} +C}{A},
\end{align*}
with the understanding that $\gamma_0=1$.
Hence, the recurrence relation for $\gamma_i$ gives
$$
\gamma_i = \left( 1-\frac{y}{x-1}\right) x^i + \frac{y}{x-1}, ~~i=1\ldots, t.
$$
The last expression for $\gamma_i$, when $i=t$ should agree with~\eqref{equa: opt gt}.
It is straightforward to see that since $R\geq 3$, we obtain that $x>1>0$. So condition $\gamma_i >\gamma_{i-1}$ translates into that $x-y>1$, which also guarantees that $\gamma_1>1$. Finally, the last condition asserts that $\gamma_t < \beta$.
\end{proof}

Theorem~\ref{thm: best t-sub-origin} suggests that in order to obtain an efficient $t$-sub-monotone algorithm with parameters $\beta, \gamma_1, \ldots, \gamma_t$, we need to minimize $R$ subject to constraint~\eqref{equa: nonlinearequality} (and to the associated strict inequality constraints). Ideally, we would like to find all roots to the associated (at least) degree-$t$ polynomial in $R$, and identify the minimum root that complies with the remaining feasibility conditions. The task is particularly challenging (from a numerical perspective), since that polynomial's coefficients depend also on the unknown value $\beta$. To bypass this difficulty, and for fixed $p,t$, we define intuitive values of $R,\beta$ that always satisfy the constraint, for which we need to check separately that they induce valid search trajectories (which is established by checking the two strict inequalities). Numerical calculations suggest that this heuristic choice of $R,\beta$ is the optimal one, but a proof is eluding us. Nevertheless, the choice of $R,\beta$ is valid, which is summarized by the statement of Theorem~\ref{thm: best t-sub-origin} and which we are ready to prove next.

\medskip

\noindent \emph{Proof of Theorem}~\ref{thm: best t-sub-origin}.
Expression~\eqref{equa: nonlinearequality} is a rational function on $\beta$. Tedious (and software assisted symbolic calculations) show that the numerator of that rational function is the $t$-characteristic polynomial $q_0+q_1\beta+q_2\beta^2$ of pair $(p,R)$.
If $R$ is such that the discriminant of that polynomial is equal to 0, then $-q_1/2q_2$ is a root to the polynomial, and hence, constraint~\eqref{equa: nonlinearequality} is satisfied for the values of $p,R,\beta,t$.
Since pair $(\beta,R)$ is feasible,
all preconditions of Lemma~\ref{lem: best t-sub-origin} are satisfied, and hence, $\beta,\gamma_1, \ldots, \gamma_t$ is a $t$-sub-monotone algorithm with competitive ratio $R$ for problem \FS. 
\medskip

We observe that Theorem~\ref{thm: best t-sub-origin} computes exactly the best monotone algorithm of Lemma~\ref{lem: upper bound backtoorigin}.  In other words, the $0$-sub-monotone we propose above is the optimal monotone algorithm we have already studied. Indeed, the discriminant of the $0$-characteristic polynomial of $(p,R)$ equals
$$
(p-2)^2 p^2 (p (p (17-2 p ((p-6) p+12))+R)-2 R)^2 \left((p+R) \left((p-2)^2 R+p ((p-4) p+12)\right)-16 R\right)
$$
The two roots of the right-hand-side factor above is a degree-2 polynomial in $R$ with roots $\frac{4\pm4\sqrt{1-p}}{2-p}-p$, one of which (the only one which is at least $3$) being exactly the competitive ratio calculated by Lemma~\ref{lem: upper bound backtoorigin}. Moreover, setting $\beta=-q_1/2q_2$ gives the same value of the expansion factor, which is denoted by $b$ in Lemma~\ref{lem: upper bound backtoorigin}.

\section{Theorem~\ref{thm: heuristic t=1}}
\begin{proof}
We fix $\beta=1/(1-p)$ and invoke constraint~\eqref{equa: nonlinearequality}, so as to force that the competitive ratio does not depend on which subinterval the treasure is placed within a 1-Hop of the $1$-sub-monotone-algorithm. The constraint then becomes
$$
\frac{1}{2} \left(\frac{R}{p-p^2}+\frac{p-4}{p^2-3 p+2}-\frac{4 ((p-1) p+2) (p-2)^2}{p (p (2 p-9)-R+12)+2 (R-4)}\right)=0,
$$
which, solved for $R$, gives the promised competitive ratio.

As for the turning point $\gamma_1$ of the $1$-Hop, it can be computed as $\frac{E}{R/p-F}$ and in order to be valid, it has to be positive and at most $\beta=1/(1-p)$. This is verified in Figure~\ref{fig: heuristic t=1,2}-middle.
\end{proof}

\end{showappendix}

\end{document}